\documentclass[%
 aip,
 jmp,
 amsmath,amssymb,
 reprint, onecolumn
]{revtex4-1}

\usepackage{graphicx}% Include figure files
\usepackage{dcolumn}% Align table columns on decimal point
\usepackage{bm}% bold math
\usepackage[utf8]{inputenc}
\usepackage[T1]{fontenc}
\usepackage{mathptmx}
\usepackage{etoolbox}
\usepackage{dsfont}% \mathds{1} for the identity operator

\usepackage{amsthm}
\usepackage{enumitem}
\usepackage{mathtools}
\usepackage{braket}
\usepackage{tikz}
\usetikzlibrary{arrows.meta}
\usepackage{hyperref}
\hypersetup{colorlinks=true,linkcolor=blue,citecolor=blue,urlcolor=blue,
  pdftitle={Galilean Reeh-Schlieder obstruction in the vacuum and the thermal Reeh-Schlieder property},
  pdfauthor={Leonardo A. Pach\'on}}

\newtheorem{theorem}{Theorem}
\newtheorem{proposition}{Proposition}
\newtheorem{lemma}{Lemma}
\newtheorem{corollary}{Corollary}
\newtheorem{alemma}{Lemma}

\theoremstyle{definition}
\newtheorem{definition}{Definition}
\theoremstyle{remark}
\newtheorem{remark}{Remark}

\makeatletter
\def\@email#1#2{%
 \endgroup
 \patchcmd{\titleblock@produce}
  {\frontmatter@RRAPformat}
  {\frontmatter@RRAPformat{\produce@RRAP{*#1\href{mailto:#2}{#2}}}\frontmatter@RRAPformat}
  {}{}
}%
\makeatother

\newcommand{\R}{\mathbb{R}}
\newcommand{\C}{\mathbb{C}}
\newcommand{\Z}{\mathbb{Z}}
\newcommand{\N}{\mathbb{N}}
\newcommand{\F}{\mathcal{F}}
\newcommand{\A}{\mathcal{A}}
\newcommand{\M}{\mathcal{M}}
\newcommand{\Rg}{\mathcal{R}}
\newcommand{\Hil}{\mathcal{H}}
\newcommand{\K}{\mathcal{K}}
\newcommand{\Op}{\mathcal{O}}
\newcommand{\B}{\mathcal{B}}
\newcommand{\Fock}{\Gamma_{\rm s}}
\newcommand{\one}{\mathds{1}}
\newcommand{\rmd}{\mathrm{d}}
\newcommand{\rmi}{\mathrm{i}}
\newcommand{\rme}{\mathrm{e}}
\newcommand{\bv}{\mathbf{v}}
\newcommand{\ba}{\mathbf{a}}
\newcommand{\bx}{\mathbf{x}}
\newcommand{\by}{\mathbf{y}}
\newcommand{\bk}{\mathbf{k}}
\newcommand{\hH}{\hat H}
\newcommand{\hM}{\hat M}
\newcommand{\hP}{\hat{\mathbf P}}
\newcommand{\hN}{\hat N}
\newcommand{\psio}{\hat\psi_0}
\newcommand{\psiod}{\hat\psi_0^\dagger}
\newcommand{\evol}[1]{\rme^{\rmi #1\hH/\hbar}}
\newcommand{\evolm}[1]{\rme^{-\rmi #1\hH/\hbar}}
\DeclareMathOperator{\supp}{supp}
\DeclareMathOperator{\Tr}{Tr}

\begin{document}

\preprint{}

\title{Galilean Reeh--Schlieder obstruction in the vacuum and the thermal
Reeh--Schlieder property}

\author{Leonardo A.\ Pach\'{o}n}
\affiliation{guane Enterprises, R+D+I Unit, Medell\'{i}n 050010, Colombia}
\email{leonardo@guane.com.co}

\date{\today}

\begin{abstract}
We ask in which states, and for which local algebras, Galilean
quantum fields have the Reeh--Schlieder property. The nets are generated
by the algebras of time-zero canonical fields on spatial regions and by
the dynamics. The Bargmann mass divided by the charge of the field is a
number operator for the time-zero fields, positivity of the energy and
Galilean boost covariance bound it from below, and Chaiken's characterisation of such operators gives
Galilean Fock rigidity: the time-zero fields form a direct sum of Fock
representations, and every translation-invariant vector is annihilated
by the time-zero annihilation fields and generates a Fock
representation. If the vacuum is the only vector invariant under spatial
translations, a condition that the usual uniqueness axiom does not
imply, the vacuum sector of the observable net is one-dimensional. No
vector of finite mass is cyclic or separating for an algebra of a
spatial region at fixed time. In the vacuum of Bose and Fermi gases with pair potentials that are
bounded, of Coulomb type, or repulsive and arbitrarily singular at the
origin, also in external
potentials such as traps, and with Buchholz's renormalised dynamics if
the interaction is not H-stable, the fields of any region at the times
of any interval generate the algebra of all bounded operators, although
at bounded energy free one-particle propagation is local up to errors
that decay faster than any power of the distance beyond the maximal
velocity. By contrast, a KMS state is cyclic and separating for the
algebra of a time tube $B\times I$ if and only if the spatial algebras of
$B$ at all times generate the global algebra. This holds for ideal gases,
also in external potentials, where in a harmonic trap the global algebra
is of type I, and it survives bounded interactions localised in $B$. For the free Schr\"odinger field we classify the
vacuum, the thermal states of Bose and Fermi gases and the Fermi sea:
fixed-time field algebras are type I factors in all of them, thermal vectors and the Fermi sea are separating
but not cyclic for them, and every field algebra with temporal extent is
the global one, of type I in pure states and of type III$_1$ in thermal
states, where its modular group is the grand-canonical dynamics.
Galilean vacua are thus never separating for local algebras. For
the free field the type of the algebras of regions with temporal extent
depends on the state, and what is specifically relativistic is the
locality of type III$_1$.
\end{abstract}

\maketitle

% =====================================================================
\section{Introduction}\label{sec:intro}
% =====================================================================

\subsection{The problem}

In relativistic quantum field theory the vacuum is cyclic and separating
for the algebra of every bounded open region~\cite{ReehSchlieder1961,
Haag1996}.
This Reeh--Schlieder property underlies the entanglement of the vacuum
across every surface, the Tomita--Takesaki modular theory of local
algebras~\cite{Takesaki1970} and the fact that these algebras are type
III$_1$ factors~\cite{BuchholzDAntoniFredenhagen1987,Verch1997}. Its
proof uses the spectrum condition together with Lorentzian
locality~\cite{Haag1996,Halvorson2001}.

Galilean quantum field theories are formulated with local commutativity
only at equal times~\cite{LevyLeblond1967}: every pair of non-simultaneous events of
Galilean spacetime is causally connectible, Galilean covariance excludes
unequal-time locality of canonical fields realised at short times, and
with positive energy it confines unequal-time commutation to null sets of
time separations~\cite{Pachon2026a}. What survives of the Reeh--Schlieder property
has been answered in parts. Requardt showed that the vacuum is cyclic
for algebras of spacetime regions with temporal
extent~\cite{Requardt1982,Requardt1986}. Bain observed that it is not
separating for them, and that it is neither cyclic nor separating for
algebras of spatial regions at a fixed time~\cite[\S~5.1]{Bain2011}.
Papageorgiou and Pye, studying the non-relativistic regime of the
Klein--Gordon field, left open whether the Reeh--Schlieder theorem has
an analogue there~\cite[Sec.~6]{PapageorgiouPye2019}.
The companion paper~\cite{Pachon2026a} shows that, under Bargmann
covariance and positivity of the energy, the vacuum is annihilated by
the time-zero annihilation fields and is not separating for fixed-time
algebras (Lemma~5 and Corollary~2 there).

We ask in which states, and for which algebras, the property holds. The
answer is a dichotomy between the vacuum and thermal states.

\subsection{Results}

\paragraph{Fock rigidity and the vacuum.} In a Galilean theory whose
time-zero canonical fields integrate to a Weyl system, the mass operator
divided by the Bargmann charge of the field is a number operator in
Chaiken's sense~\cite{Chaiken1968}, and it is bounded from below by
positivity of the energy and boost covariance~\cite{Pachon2026a}. Hence
the Weyl system is a direct sum of Fock representations, and every
translation-invariant vector is annihilated by the time-zero annihilation
fields and generates a Fock representation of them
(Theorem~\ref{thm:rigidity}). If the vacuum is the only vector invariant
under spatial translations, the vacuum sector of the observable net is
one-dimensional (Theorem~\ref{thm:vacuum}). This condition does not
follow from the usual uniqueness of the vector invariant under
space-time translations (Proposition~\ref{prop:counterexample}).
\paragraph{Finite mass and time tubes.} For irreducible time-zero fields
the observable algebras of spatial regions at fixed time carry a local
mass superselection in every sector of finite mass, and no vector of
finite mass is cyclic or separating for the field algebras of such
regions (Theorem~\ref{thm:mass}). For Bose and Fermi gases with pair interactions, the algebra of every
region $B\times I$ with temporal extent is $\B(\Hil)$ in the vacuum
representation, for the physical dynamics if the interaction is H-stable
and for Buchholz's renormalised dynamics otherwise
(Theorem~\ref{thm:tube}), as it is for the free
field~\cite{Pachon2026c}. The pair potential may be bounded, of Coulomb
type, or repulsive and arbitrarily singular at the origin, and the gas
may be placed in an external potential such as a trap or the Coulomb
field of fixed charges. Without external potential, the density property
that the proof needs (Lemma~\ref{lem:UC}) follows from Requardt's
analysis of $n$-body Schr\"odinger operators~\cite{Requardt1986}; we
prove it through Masuda's unique continuation theorem for the
Schr\"odinger evolution~\cite{Masuda1967}, in the form of Lemma~A7 of
Ref.~\onlinecite{Pachon2026a}, which needs no translation invariance. At
bounded energy, reproducing from a
time tube a particle localised beyond the reach of the maximal velocity
requires local smearing functions that grow faster than any power of the
distance, although for the free field they cancel so well that the local
operator itself can have norm of order one
(Proposition~\ref{prop:bounded-energy}).
\paragraph{Thermal states.} A KMS state is cyclic and separating for the
algebra of a time tube $B\times I$ if and only if the spatial algebras of
$B$ at all times generate the global algebra
(Theorem~\ref{thm:thermalRS}). Sufficiency is the Reeh--Schlieder
argument with the KMS strip in place of the spectrum condition, and
necessity follows from Takesaki's theorem on conditional expectations.
This tube completeness holds for ideal gases in external potentials; in a
harmonic trap the global algebra of a thermal state is of type I, and
with continuous one-particle spectrum it is of type III$_1$
(Proposition~\ref{prop:external}). It survives bounded interactions
localised in $B$, and so does the thermal Reeh--Schlieder property of the
corresponding KMS states (Proposition~\ref{prop:local}).
\paragraph{The free field.} For the free Schr\"odinger field we classify
the vacuum, the thermal states of Bose and Fermi gases and the Fermi sea
(Theorem~\ref{thm:quasifree}, Table~\ref{tab:classification}). The
algebra of every region with temporal extent is the global algebra, of
type I in pure states and of type III$_1$ in thermal states, where its
modular group is the grand-canonical dynamics. The type of the algebras
of regions with temporal extent therefore depends on the state
(Corollary~\ref{cor:type}); modular crossed products of fixed-time
algebras are of type I, while in thermal states the algebra of every
time tube has the type II$_\infty$ crossed product of the global
algebra (Proposition~\ref{prop:crossed}); and what is specifically
relativistic is the locality of type III$_1$
(Corollary~\ref{cor:locality}).

\subsection{Relation to previous work}

\paragraph{Axioms and the vacuum.} Dadashev formulated Wightman-type
axioms for Galilean fields with non-negative particle-number operators,
showed that the fields annihilate the vacuum and that canonical
time-zero fields for which the vacuum is cyclic are Fock
fields~\cite{Dadashev1985}. Buchholz constructed the net of observables
of non-relativistic Bose fields from the resolvent algebra, with
dynamics for pair interactions~\cite{Buchholz2018}, and noted that the
vacuum representation of the observable algebra is one-dimensional
because the gauge generator is bounded from
below~\cite{Buchholz2020}. Dell'Antonio, Doplicher and Ruelle
characterised the direct sums of Fock representations by the existence
of a total number operator~\cite{DellAntonioDoplicherRuelle1966}, and
Chaiken showed that a number operator whose spectrum is bounded below
already forces such a direct sum~\cite{Chaiken1968}. What
Dadashev and Buchholz assume, the non-negativity of the particle number,
follows here from positivity and boost covariance through the companion
paper~\cite{Pachon2026a}; L\'evy-Leblond had observed that a
Galilean-invariant vacuum has zero mass~\cite{LevyLeblond1967}. Puccini
and Vucetich gave axioms admitting masses of both
signs~\cite{PucciniVucetich2004}.

\paragraph{Reeh--Schlieder.} The cyclicity of the vacuum for regions
with temporal extent is due to Requardt~\cite{Requardt1982}, and its
non-separation to Bain~\cite[\S~5.1]{Bain2011}. Requardt also proved
density theorems of Reeh--Schlieder type for Schr\"odinger
operators~\cite{Requardt1986}. For operators with the unique continuation
property and a generalised eigenfunction expansion, which by Agmon's
theorem exists for short-range potentials, the vectors $\rme^{-\rmi tH}
\Phi$ with $\Phi$ supported in a given open set and $t$ in a given
interval span a dense subspace (his Theorem~2). For $n$ particles with
translation-invariant pair potentials of relative bound less than one,
so do the multi-time states obtained by alternately multiplying with
functions supported in a given open set and evolving with the
$k$-particle Hamiltonians over differences of times in the interval
(his Theorem~6); he proves this by analytic continuation in the times
and positions of the individual particles. For the free field the algebra of such a
region is $\B(\Hil)$ (Ref.~\onlinecite{Pachon2026c}, Proposition~3);
Theorem~\ref{thm:tube} extends this to interacting Bose and Fermi gases;
without external potential its density lemma (Lemma~\ref{lem:UC})
follows from Requardt's Theorem~6 and its proof, and in external
potentials it lies outside his hypotheses. For relativistic thermal states the Reeh--Schlieder property
was proved by J\"akel~\cite{Jaekel2000}, and on stationary spacetimes
Strohmaier proved it for quasi-free ground and KMS states, together with
an analogue of the timelike tube theorem, by combining analyticity in
the KMS strip with unique continuation~\cite{Strohmaier2000}. Falcone and Conti argued that
the non-locality associated with the Reeh--Schlieder theorem is
suppressed under non-relativistic conditions~\cite{FalconeConti2024}.
The non-relativistic limit of the Klein--Gordon field at the level of
local algebras is treated in Ref.~\onlinecite{Pachon2026c}.

\paragraph{Quasi-free states.} The representations of the free Bose and
Fermi gases are those of Araki and Woods and of Araki and
Wyss~\cite{ArakiWoods1963,Derezinski2006}, and the lattice of their local
algebras is Araki's~\cite{Araki1963,Derezinski2006}. The local normality
that makes fixed-time algebras type I is the condition of
Chaiken and of Powers and St\o rmer~\cite{Chaiken1968,PowersStormer1970}.
The confined gases, whose thermal states are Gibbs states on the Fock
space, are treated in Ref.~\onlinecite{Derezinski2006} as well. Araki
constructed the KMS states of bounded perturbations of a dynamics with a
faithful normal KMS state~\cite{Araki1973}.

\paragraph{What is new.} Part~(i) of Theorem~\ref{thm:rigidity}, which
combines Ref.~\onlinecite{Pachon2026a} with Chaiken's theorem and
replaces the Fock axiom by a consequence of positivity and covariance,
is recorded after Lemma~5 of Ref.~\onlinecite{Pachon2026a}. Part~(iii),
which locates the translation-invariant vectors, and the identification
$\hM_\K\xi=0$ in~(iv) are new; part~(ii) is the standard cocycle
argument, and part~(v) recovers Dadashev's Fock property with the
positivity of the particle number derived rather than
assumed~\cite{Dadashev1985}. Proposition~\ref{prop:counterexample} is
new:
the uniqueness of the spatially invariant vector, which
Theorem~\ref{thm:vacuum} needs, cannot be weakened to the usual
uniqueness axiom. Theorem~\ref{thm:mass} follows directly from
Theorem~\ref{thm:rigidity}. For the free field the algebra of a time
tube is $\B(\Hil)$ by Proposition~3 of Ref.~\onlinecite{Pachon2026c}.
Theorem~\ref{thm:tube} extends this to interacting Bose and Fermi
gases, with the renormalised dynamics when the interaction is not
H-stable. Without external potential the density on which it rests is
Requardt's (Ref.~\onlinecite{Requardt1986}, Theorem~6 and its proof); in
external potentials, which break translation invariance, it is not
covered by his theorems. In all cases the theorem identifies the algebra
through the vacuum projection. Proposition~\ref{prop:bounded-energy}
quantifies the locality at bounded energy. Theorem~\ref{thm:thermalRS}
is the Reeh--Schlieder induction with the KMS strip in place of the
spectrum condition, as in Strohmaier's treatment of stationary
spacetimes~\cite{Strohmaier2000}; that tube completeness is also
necessary follows from Takesaki's theorem~\cite{Takesaki1972}.
Proposition~\ref{prop:local} combines Araki's perturbation
theory~\cite{Araki1973} with the Dyson expansion of the cocycle that
relates the two dynamics, and is new to our knowledge. In
Proposition~\ref{prop:external}, (a) is, for short-range potentials,
Requardt's Theorem~2~\cite{Requardt1986}, and for the potentials of
static stars it is contained in the proof of Proposition~3 of
Ref.~\onlinecite{Pachon2026c}; the type I of the confined gas in~(d) and
the type III$_1$ in~(e) are known~\cite{Derezinski2006}. The equality of
tube and global algebras in~(b) at non-zero density, and with it the
thermal Reeh--Schlieder property of time tubes in external potentials,
also with a global algebra of type I, are new to our knowledge. In Theorem~\ref{thm:quasifree}, the type III$_1$ of
the global algebra, the cyclic and separating vector and the modular
operator of the representations of Araki and Woods and of Araki and
Wyss, and the
type I of fixed-time algebras are
known~\cite{ArakiWoods1963,Derezinski2006,Chaiken1968,PowersStormer1970};
the equality of tube and global algebras, the separating but not cyclic
vector at fixed time, including the Fermi sea, and the classification of
Table~\ref{tab:classification} are new to our knowledge, as are the
consequences in Sec.~\ref{sec:consequences}. They answer the question
of Ref.~\onlinecite{PapageorgiouPye2019} as follows: in KMS states the
Reeh--Schlieder property holds for a time tube $B\times I$ if and only if
tube completeness holds for $B$, and then for every region containing
the tube, as it does for ideal gases, also in external potentials, and,
for the perturbed dynamics, for every $B$ in which a bounded interaction
is localised; it fails in the
vacuum of every Galilean net; and for regions without temporal extent it
fails in the vacuum and in the thermal states of the free gases.

The paper is organised as follows. Section~\ref{sec:setting} states the
framework and gives an interacting example. Section~\ref{sec:rigidity}
proves Fock rigidity and describes the vacuum sector.
Section~\ref{sec:mass} treats finite-mass sectors, time tubes in the
vacuum, also for Fermi gases and in external potentials, and locality at
bounded energy, Sec.~\ref{sec:thermal} thermal states and their bounded
local perturbations, and Sec.~\ref{sec:free} the free field and ideal
gases in external potentials. Section~\ref{sec:consequences} draws the consequences
for the type of local algebras and modular crossed products, and
Sec.~\ref{sec:discussion} lists open problems. Appendix~\ref{app:aux}
collects auxiliary lemmas. Figure~\ref{fig:structure} shows how the
results depend on each other.

\begin{figure}[htb]
\centering
\begin{tikzpicture}[font=\footnotesize, >={Stealth[length=2mm]},
  box/.style={rounded corners=2pt, align=center, inner sep=3pt,
    minimum height=9mm, text width=38mm,
    execute at begin node={\hyphenpenalty=10000\exhyphenpenalty=10000}},
  hyp/.style={box, draw=gray!70!black, dashed, fill=gray!10},
  der/.style={box, draw, fill=green!12},
  thm/.style={box, draw, line width=0.8pt, fill=blue!12, font=\footnotesize\bfseries},
  wide/.style={text width=42mm},
  arr/.style={->, gray!70!black, line width=0.5pt}]
\def\figXa{2.3}\def\figXb{6.8}\def\figXc{10.95}\def\figXd{15.1}
\def\figYa{0}\def\figYb{-1.6}\def\figYc{-3.2}\def\figYd{-4.75}\def\figYe{-6.3}\def\figYf{-7.8}
\node[hyp, wide] (ax) at (\figXa,\figYa)
  {Axioms (G1)--(G6); for Theorems~\ref{thm:vacuum} and~\ref{thm:mass} also
   (G6$'$) or irreducibility};
\node[der, wide] (num) at (\figXa,\figYb)
  {(G3): $\hM/m_0$ is a number operator; (G5) and
   Ref.~\onlinecite{Pachon2026a}, Lemma~4: $\hM\geq0$};
\node[thm, wide] (t1) at (\figXa,\figYc)
  {Theorem~\ref{thm:rigidity}\\ \mdseries Fock rigidity};
\node[thm, wide] (t23) at (\figXa,\figYd)
  {Theorems~\ref{thm:vacuum} and~\ref{thm:mass},
   Proposition~\ref{prop:zero-mass}\\ \mdseries vacuum sector, finite mass};
\node[hyp] (bose) at (\figXb,\figYa)
  {Bose and Fermi gases, pair and external potentials~(P)};
\node[der] (uc) at (\figXb,\figYb)
  {Lemma~\ref{lem:UC}: unique continuation (Masuda; for $U=0$ also Requardt)};
\node[thm] (t4) at (\figXb,\figYd)
  {Theorem~\ref{thm:tube}\\ \mdseries $\F(B\times I)=\B(\Hil)$ in the vacuum};
\node[thm] (p4) at (\figXb,\figYe)
  {Proposition~\ref{prop:bounded-energy}\\ \mdseries locality at bounded energy};
\node[der] (ind) at (\figXc,\figYb)
  {Lemma~\ref{lem:induction}: induction in the tube};
\node[der] (app) at (\figXc,\figYe)
  {Lemmas~\ref{lem:real-subspaces}--\ref{lem:gap}: real subspaces, density,
   spectral gap};
\node[hyp] (kms) at (\figXd,\figYa) {KMS state};
\node[thm] (t5) at (\figXd,\figYc)
  {Theorem~\ref{thm:thermalRS}\\ \mdseries thermal Reeh--Schlieder
   $\Leftrightarrow$~\eqref{eq:tube-completeness}};
\node[thm] (p5) at (\figXc,\figYd)
  {Proposition~\ref{prop:local}\\ \mdseries bounded local interactions};
\node[thm] (p6) at (\figXd,\figYd)
  {Proposition~\ref{prop:external}, with Lemma~\ref{lem:UC}\\
   \mdseries ideal gases in external potentials};
\node[thm] (t6) at (\figXd,\figYe)
  {Theorem~\ref{thm:quasifree}, Table~\ref{tab:classification}\\
   \mdseries the free field};
\node[thm, text width=56mm] (cor) at (13.0,\figYf)
  {Corollaries~\ref{cor:type} and~\ref{cor:locality},
   Proposition~\ref{prop:crossed}\\
   \mdseries type, crossed products, locality of type III$_1$};
\draw[arr] (ax) -- (num);
\draw[arr] (num) -- node[right, font=\footnotesize, black] {Chaiken} (t1);
\draw[arr] (t1) -- (t23);
\draw[arr] (bose) -- (uc);
\draw[arr] (uc) -- (t4);
\draw[arr] (ind) -- (t4.north east);
\draw[arr] (ind) -- (t5);
\draw[arr] (kms) -- (t5);
\draw[arr] (app) -- (p4);
\draw[arr] (app) -- (t6);
\draw[arr] (t6) -- (p6);
\draw[arr, dashed] (p6) -- node[right, font=\footnotesize, black]
  {\eqref{eq:tube-completeness}} (t5);
\draw[arr] (t5) -- (p5);
\draw[arr] (t6) -- (t6 |- cor.north);
\end{tikzpicture}
\caption{Main dependencies between the results. Dashed boxes:
hypotheses; green: lemmas and derived facts; bold boxes: results.
Chaiken's theorem, that a number operator bounded below forces a direct
sum of Fock representations, turns the lower bound on $\hM/m_0$ into
Fock rigidity. Theorems~\ref{thm:vacuum} and~\ref{thm:mass} need in
addition (G6$'$) or irreducible time-zero fields, which (G6) does not
imply (Proposition~\ref{prop:counterexample}). Lemma~\ref{lem:induction}
passes from all times to a time interval, with the spectrum condition in
the vacuum and the KMS strip in thermal states. For $U=0$, Lemma~\ref{lem:UC} follows from Requardt's density
theorem~\cite{Requardt1986} and its proof; we prove it through Masuda's
unique continuation theorem~\cite{Masuda1967}, which also covers
external potentials. Theorem~\ref{thm:tube} does not use the axioms: it
holds for the gases of~(P), most of which are not Galilean nets, and its
statement on the observable algebras repeats the argument of
Theorem~\ref{thm:mass}(iii). Proposition~\ref{prop:bounded-energy}
quantifies the locality of the tube algebras at bounded energy; it uses
Lemma~\ref{lem:density}, not Theorem~\ref{thm:tube}.
Proposition~\ref{prop:external}, which extends the proof of
Theorem~\ref{thm:quasifree}(ii) to external potentials, $U=0$ included,
and uses Lemma~\ref{lem:UC} with one particle, proves tube
completeness~\eqref{eq:tube-completeness} for ideal gases (dashed arrow);
in a harmonic trap the global algebra is of type I.
Proposition~\ref{prop:local} applies Theorem~\ref{thm:thermalRS} to
bounded interactions localised in $B$, to which it transfers tube
completeness. Corollary~\ref{cor:type} also uses Theorem~\ref{thm:tube},
for the gases of~(P).}
\label{fig:structure}
\end{figure}

% =====================================================================
\section{Galilean nets}\label{sec:setting}
% =====================================================================

\subsection{Bargmann representations and time-zero fields}

We use the unitary Bargmann representations of
Ref.~\onlinecite{Pachon2026a}, Sec.~IV.A: strongly continuous unitary
groups $\evol{t}$ (time translations), $T(\ba)=\rme^{-\rmi\ba\cdot\hP/\hbar}$
(spatial translations), $B(\bv)$ (boosts), $U(R)$ (rotations) and
$\rme^{\rmi\theta\hM/\hbar}$ (the central mass group), subject to the
Weyl form (B1)--(B5) there of the relations $[\hat K_j,\hat P_k]=
\rmi\hbar\delta_{jk}\hM$ and $[\hat K_j,\hH]=\rmi\hbar\hat P_j$. In
generator form $B(\bv)^{-1}\hP B(\bv)=\hP+\hM\bv$ and $B(\bv)^{-1}\hH
B(\bv)=\hH+\bv\cdot\hP+\frac12|\bv|^2\hM$. For a function $g$ on $\R^3$ we
write $g_\ba:=g(\cdot-\ba)$ and $e_\bv(\bx):=\rme^{\rmi m_0\bv\cdot
\bx/\hbar}$, and $\braket{f,g}=\int\overline{f}g\,\rmd^3x$.

A \emph{Weyl system} over $L^2(\R^3)$ on a Hilbert space $\Hil$ is a map
$g\mapsto W(g)$ into the unitary operators with
\begin{equation}\label{eq:weyl}
W(g_1)W(g_2)=\rme^{-\rmi\hbar\,\mathrm{Im}\braket{g_1,g_2}/2}\,W(g_1+g_2),
\end{equation}
which is \emph{regular}: $t\mapsto W(tg)$ is strongly continuous for
every $g$. Writing $W(g)=\rme^{\rmi\Phi(g)}$ with $\Phi(g)$ self-adjoint,
the time-zero annihilation and creation fields are the closures of
$\psio(g)=2^{-1/2}(\Phi(g)+\rmi\Phi(\rmi g))$ and $\psiod(g)=2^{-1/2}
(\Phi(g)-\rmi\Phi(\rmi g))$; $\psio$ is antilinear, $\psiod$ is linear,
and $[\psio(f),\psiod(g)]=\hbar\braket{f,g}$ on a common core. The Fock
representation $W_{\rm F}$ on $\Fock(L^2(\R^3))$, with vacuum
$\Omega_{\rm F}$ and number operator $\hN_{\rm F}$, is the one in which
$\psio(g)\Omega_{\rm F}=0$ for all $g$.

\begin{definition}[Galilean net]\label{def:net}
A \emph{Galilean net} $(\Hil,U,W,\Omega)$ consists of a Hilbert space
$\Hil$, a unitary Bargmann representation $U$, a Weyl system $W$ over
$L^2(\R^3)$ and a unit vector $\Omega\in\Hil$ such that:
\begin{enumerate}[label=(G\arabic*)]
\setcounter{enumi}{2}
\item \emph{Covariance with Bargmann charge.} For some $m_0>0$, all
$g\in L^2(\R^3)$ and all group parameters,
\begin{gather*}
T(\ba)W(g)T(\ba)^{-1}=W(g_\ba),\qquad U(R)W(g)U(R)^{-1}=W(g\circ R^{-1}),
\qquad B(\bv)W(g)B(\bv)^{-1}=W(e_\bv g),\\
\rme^{\rmi\theta\hM/\hbar}W(g)\rme^{-\rmi\theta\hM/\hbar}=
W(\rme^{\rmi\theta m_0/\hbar}g).
\end{gather*}
\item \emph{Canonical fields.} $W$ is regular.
\item \emph{Positive energy.} $\hH\geq0$.
\item \emph{Vacuum.} $\Omega$ is invariant under $T(\ba)$, $\evol{t}$ and
$U(R)$, and it is, up to a phase, the only vector invariant under
$T(\ba)$ and $\evol{t}$ for all $\ba$ and $t$.
\end{enumerate}
We say that the net satisfies (G6$'$) if $\Omega$ is, up to a phase,
the only vector invariant under the spatial translations $T(\ba)$.
\end{definition}

For a non-empty bounded open $B\subset\R^3$ and $t\in\R$, the
\emph{slice algebra} and its time translate are
\begin{equation}\label{eq:slices}
\F(B,0):=\{W(g):g\in L^2(B)\}'',\qquad
\F(B,t):=\evol{t}\,\F(B,0)\,\evolm{t},
\end{equation}
and for an open set $\Op\subset\R^3\times\R$ we put
\begin{equation}\label{eq:regions}
\F(\Op):=\bigvee\bigl\{\F(B,t):\ B\times\{t\}\subset\Op\bigr\}.
\end{equation}
By construction the net satisfies
\begin{enumerate}[label=(G\arabic*)]
\item \emph{Isotony:} $\Op_1\subset\Op_2$ implies $\F(\Op_1)\subset
\F(\Op_2)$, and likewise for slices;
\item \emph{Equal-time locality:} $[\F(B_1,t),\F(B_2,t)]=\{0\}$ if
$B_1\cap B_2=\emptyset$,
\end{enumerate}
the latter because $W(g_1)$ and $W(g_2)$ commute by~\eqref{eq:weyl} when
$g_1$ and $g_2$ have disjoint supports. The algebras of four-dimensional
regions are not required to commute. Under the hypotheses of
Corollary~4 of Ref.~\onlinecite{Pachon2026a}, which include time-smeared
fields affiliated with them, Galilean covariance and positivity of the
energy exclude it. The fields $\psio(g)$,
$\psiod(g)$, $g\in C_c^\infty(B)$, are affiliated with $\F(B,0)$. The
\emph{observable algebras} are the gauge-invariant parts
\begin{equation}\label{eq:observables}
\A(\Op):=\{X\in\F(\Op):\ \rme^{\rmi\theta\hM/\hbar}X
\rme^{-\rmi\theta\hM/\hbar}=X\ \text{for all }\theta\in\R\},
\end{equation}
and likewise $\A(B,t)$.

\begin{remark}[Relation to other formulations]\label{rem:axioms}
The labels (G1)--(G6) follow the axioms of Galilean Haag--Kastler nets
used in Ref.~\onlinecite{Pachon2026c}. The present formulation defines
the net
by its slices and the dynamics, as in
Refs.~\onlinecite{Pachon2026a,Buchholz2018}; smeared four-dimensional
fields and their time-zero limits are not needed. The Bargmann charge
condition in (G3) is the Weyl form of $\rme^{\rmi\theta\hM/\hbar}\psio(g)
\rme^{-\rmi\theta\hM/\hbar}=\rme^{-\rmi\theta m_0/\hbar}\psio(g)$. The
Fock representation, an axiom in earlier formulations, is a consequence
here (Theorem~\ref{thm:rigidity}). Theorems~\ref{thm:rigidity}--\ref{thm:mass} and
Propositions~\ref{prop:zero-mass} and~\ref{prop:counterexample} hold
verbatim under positivity modulo the central charge,
$\hH+\lambda\hM\geq0$ for some $\lambda\in\R$, because they use (G5)
only through Lemma~4 of Ref.~\onlinecite{Pachon2026a}, which is proved
under that hypothesis.
\end{remark}

\begin{definition}[Reeh--Schlieder property]\label{def:RS}
A vector $\Xi$ is \emph{cyclic} for a von Neumann algebra $\M$ if
$\M\Xi$ is dense, and \emph{separating} if $X\in\M$ and $X\Xi=0$ imply
$X=0$. A net has the \emph{Reeh--Schlieder property} for a vector $\Xi$
if $\Xi$ is cyclic and separating for the algebra of every non-empty
bounded open region.
\end{definition}

\subsection{What the axioms do not restrict}

Equal-time locality survives conjugation by any unitary group, so (G1)
and (G2) hold for every self-adjoint $\hH$; the dynamics is constrained
only by (G3), (G5) and (G6). Dadashev described the admissible Hamiltonians of
his axioms as the free one plus an invariant
perturbation~\cite{Dadashev1985}. In particular interacting examples
are easy to write down.

We use the following class of pair potentials.
\begin{enumerate}[label=(V)]
\item $V:\R^3\to\R$ is measurable and rotation invariant, $V=V_1+V_2$
with $V_1\in L^2(\R^3)$ and $V_2\in L^\infty(\R^3)$, and $V$ is bounded
on $\{|\bx|\geq\epsilon\}$ for every $\epsilon>0$.
\end{enumerate}
Bounded potentials, and Coulomb and Yukawa potentials of either sign,
satisfy~(V).

\begin{proposition}[Bose gas with pair interaction]\label{prop:pair}
Let $V$ satisfy~(V). On $\Fock(L^2(\R^3))=\bigoplus_n\Hil_n$ let $\hH_V$
act on $\Hil_n$ as $\sum_j-\frac{\hbar^2}{2m_0}\Delta_j+\sum_{i<j}
V(\bx_i-\bx_j)$ on $\Hil_n\cap H^2(\R^{3n})$, let
$-E(n):=\inf\sigma(\hH_V\restriction\Hil_n)$, and let
$\hH:=\hH_V+E(\hN_{\rm F})$. With the second-quantised Bargmann
representation of mass $m_0$ in which $\evol{t}$ is generated by $\hH$,
the Fock Weyl system and $\Omega=\Omega_{\rm F}$, this is a Galilean net
satisfying (G6$'$), with $\inf\sigma(\hH\restriction\Hil_n)=0$ for every
$n$.
\end{proposition}

\begin{proof}
By Kato's theorem~\cite{Kato1951}, the pair potential has relative bound
zero with respect to the free Hamiltonian on $L^2(\R^{3n})$, and it
commutes with permutations because $V$ is even; so $\hH_V$ is
self-adjoint on $\Hil_n\cap H^2(\R^{3n})$ and bounded below; for bounded $V$ the lower
bound is $-n(n-1)\|V\|_\infty/2$. So $E(n)$ is finite and $\hH$ is
self-adjoint with $\hH\geq0$ and $\inf\sigma(\hH\restriction\Hil_n)=0$. Translations,
rotations and the mass group are second quantised from $L^2(\R^3)$ and
commute with $\hH$, because $V$ depends only on the differences
$\bx_i-\bx_j$ and is rotation invariant. The boost $B(\bv)$ acts on
$\Hil_n$ by multiplication with $\prod_je_\bv(\bx_j)$, which commutes
with $V(\bx_i-\bx_j)$ and with $E(\hN_{\rm F})$; hence $B(\bv)^{-1}\hH
B(\bv)=\hH+\bv\cdot\hP+\frac12|\bv|^2\hM$ holds as for the free
Hamiltonian, and (B1)--(B5) of Ref.~\onlinecite{Pachon2026a} follow.
(G3) and (G4) hold for the Fock Weyl system. In $\Hil_n$, $n\geq1$, the
total momentum has absolutely continuous spectrum, so $\Omega_{\rm F}$ is
the only translation-invariant vector, and $\hH\Omega_{\rm F}=0$.
\end{proof}

The renormalisation $E(\hN_{\rm F})$ is Buchholz's~\cite{Buchholz2018}.
It adds a function of the central mass to the Hamiltonian, which leaves
the Bargmann relations and the fixed-time observable algebras $\A(B,t)$
unchanged. Unless $E$ is affine it changes, in general, the field
algebras $\F(B,t)$, $t\neq0$, because $\rme^{\rmi tE(\hN_{\rm F})/\hbar}$
multiplies $\psio(g)$ by a phase that depends on $\hN_{\rm F}$. If $V$ is
H-stable, that is, $\hH_V+\lambda\hM\geq0$ for some $\lambda\in\R$, as
for $V\geq0$ and in particular for repulsive Coulomb potentials, the renormalisation can be replaced by $\lambda\hM$,
which does not change the slice algebras
(Ref.~\onlinecite{Pachon2026a}, Remark~4); the net is then that of the
physical dynamics $\hH_V$. For attractive potentials that are not
H-stable, such as attractive Coulomb or Yukawa potentials, $E$ grows
faster than linearly, and the field net is that of the renormalised
dynamics; its fixed-time observable algebras coincide with those of
$\hH_V$.

% =====================================================================
\section{Fock rigidity and the vacuum sector}\label{sec:rigidity}
% =====================================================================

\subsection{Galilean Fock rigidity}

A self-adjoint operator $N$ is a \emph{number operator} for a Weyl system
$W$ if $\rme^{\rmi tN}W(g)\rme^{-\rmi tN}=W(\rme^{\rmi t}g)$ for all $g$
and $t$ (Ref.~\onlinecite{Chaiken1968}, Definition~1.1). Chaiken's
Weyl relations, $W(z)W(z')=\exp[\frac{\rmi}{2}\mathrm{Im}(z,z')]
W(z+z')$ with $(\cdot,\cdot)$ linear in the first argument, become
Eq.~\eqref{eq:weyl} for $z=\sqrt\hbar\,g$, so the two notions of
number operator agree, and (G4) gives the continuity of $t\mapsto
W(tg)$ that Chaiken assumes. The Fock number
operator $\hN_{\rm F}$ is one for $W_{\rm F}$, and by the last relation
in (G3) so is $\hM/m_0$ for $W$.

\begin{theorem}[Galilean Fock rigidity]\label{thm:rigidity}
Let $W$ be a regular Weyl system over a complex Hilbert space
$\mathfrak h$ on $\Hil$, with the relations~\eqref{eq:weyl}, and $\hM$ a
self-adjoint operator with $\rme^{\rmi\theta\hM/\hbar}W(g)
\rme^{-\rmi\theta\hM/\hbar}=W(\rme^{\rmi\theta m_0/\hbar}g)$ for some
$m_0>0$. Assume that $\hM$ is bounded from below; for a Galilean net
this holds, with $\hM\geq0$, by (G5) and Lemma~4 of
Ref.~\onlinecite{Pachon2026a}. Then:
\begin{enumerate}[label=(\roman*)]
\item there are a Hilbert space $\K$ and a unitary $V:\Hil\to
\Fock(\mathfrak h)\otimes\K$ with $VW(g)V^{-1}=W_{\rm F}(g)\otimes\one$ for
all $g$;
\item $V\hM V^{-1}=m_0\hN_{\rm F}\otimes\one+\one\otimes\hM_\K$ for a
self-adjoint $\hM_\K$ on $\K$, bounded from below, and non-negative if
$\hM\geq0$;
\item if $\mathfrak h=L^2(\R^3)$ and a strongly continuous unitary group
$T(\ba)$ satisfies $T(\ba)W(g)T(\ba)^{-1}=W(g_\ba)$, then every vector
$\Psi$ invariant
under all $T(\ba)$ satisfies $V\Psi\in\Omega_{\rm F}\otimes\K$; in
particular $\psio(g)\Psi=0$ for every $g\in L^2(\R^3)$;
\item for a Galilean net, $V\Omega=\Omega_{\rm F}\otimes\xi$ with
$\hM_\K\xi=0$, and $\hM\Omega=0$; in particular $W$, restricted to the
closure of $\{W(g)\}''\Omega$, is unitarily equivalent to $W_{\rm F}$, with
$\Omega$ corresponding to $\Omega_{\rm F}$;
\item if $\{W(g):g\in\mathfrak h\}''=\B(\Hil)$, then $\K\cong\C$, and for
a Galilean net, up to unitary equivalence, $W=W_{\rm F}$,
$\Omega=\Omega_{\rm F}$ and $\hM=m_0\hN_{\rm F}$.
\end{enumerate}
\end{theorem}

\begin{proof}
For a Galilean net, Lemma~4 of Ref.~\onlinecite{Pachon2026a} gives
$\sigma(\hM)\subset[0,\infty)$ from the Bargmann relations and (G5). Put $N:=\hM/m_0$, a number operator for
$W$ bounded from below.

(i) This is Theorem~1.3 of Ref.~\onlinecite{Chaiken1968}: a Weyl system
with a number operator whose spectrum is bounded below is a direct sum
of Fock representations.

(ii) Both $\rme^{\rmi tN}$ and $V^{-1}(\rme^{\rmi t\hN_{\rm F}}\otimes\one)V$
implement $W(g)\mapsto W(\rme^{\rmi t}g)$. Hence $u(t):=V\rme^{\rmi tN}
V^{-1}(\rme^{-\rmi t\hN_{\rm F}}\otimes\one)$ commutes with all
$W_{\rm F}(g)\otimes\one$, and since $W_{\rm F}$ is irreducible,
$u(t)=\one\otimes v(t)$ with $v(t)$ unitary on $\K$. Because
$\rme^{\rmi t\hN_{\rm F}}\otimes\one$ commutes with $\one\otimes v(s)$,
the group law of $\rme^{\rmi tN}$ gives $v(t+s)=v(t)v(s)$, and $v$ is
strongly continuous. Writing $v(t)=\rme^{\rmi t\hM_\K/m_0}$ gives the
formula. The subspace $\Omega_{\rm F}\otimes\K$ reduces $V\hM V^{-1}$,
which acts there as $\one\otimes\hM_\K$, so $\hM_\K$ inherits the lower
bound of $\hM$.

(iii) Both $VT(\ba)V^{-1}$ and $\Gamma(\tau_\ba)\otimes\one$, where
$\Gamma(\tau_\ba)$ is the second quantisation of $\tau_\ba g=g_\ba$,
implement $W_{\rm F}(g)\otimes\one\mapsto W_{\rm F}(g_\ba)\otimes\one$. As
in (ii), $VT(\ba)V^{-1}=\Gamma(\tau_\ba)\otimes S(\ba)$ with $S$ a strongly
continuous unitary representation of $\R^3$ on $\K$.
On $\Fock^{(n)}\otimes\K$, $n\geq1$, where $\Fock^{(n)}$ is the
$n$-particle space, the generator of $VT(\ba)V^{-1}$ is $\hP_{\rm F}\otimes
\one+\one\otimes\hP_\K$ with commuting summands. The spectral measure of
the $n$-particle momentum $\hP_{\rm F}$ is absolutely continuous, so the
spectral projection $E$ of $\hP_{\rm F}\otimes\one+\one\otimes\hP_\K$ for
$\{0\}$ satisfies $\braket{\eta\otimes\zeta,E\,\eta\otimes\zeta}=\int
\mu_\eta(\{-\mathbf p\})\,\rmd\nu_\zeta(\mathbf p)=0$ for product
vectors, where $\mu_\eta$ and $\nu_\zeta$ are the spectral measures of
$\eta$ and $\zeta$. Hence $E=0$ on $\Fock^{(n)}\otimes\K$, and an
invariant vector has no component there. Finally
$\psio(g)$ corresponds to $\sqrt\hbar\,a(g)\otimes\one$, which annihilates
$\Omega_{\rm F}\otimes\K$.

(iv) $\Omega$ is invariant under $T(\ba)$, so (iii) gives $V\Omega=
\Omega_{\rm F}\otimes\xi$. By Lemma~2 of Ref.~\onlinecite{Pachon2026a},
a vector invariant under the spatial translations has zero mass; with
(ii) this is $\hM_\K\xi=0$.

(v) The commutant of $W_{\rm F}\otimes\one$ is $\one\otimes\B(\K)$, which
is trivial only if $\K\cong\C$. Then $\hM_\K$ is a real number, which
vanishes by (iv).
\end{proof}

The theorem is the precise form of the absence of a Haag theorem for
Galilean fields~\cite{LevyLeblond1967}: positivity and boost covariance
leave only multiplicity beyond the Fock representation. Parts~(i)
and~(ii) need neither boosts nor translations once $\hM$ is bounded
below, so they apply on static backgrounds, where $\hM\geq0$ is imposed
directly (Ref.~\onlinecite{Pachon2026c}, Sec.~VI). Part~(iii) contains the
identity
$\psio(g)\Omega=0$ of Lemma~5 of Ref.~\onlinecite{Pachon2026a}, which
Dadashev proved with the non-negativity of the particle number as an
axiom~\cite{Dadashev1985}.

\subsection{The vacuum sector}

Let $\Hil_0:=\ker\hM$ be the zero-mass sector.

\begin{proposition}[Zero mass and spatial invariance]\label{prop:zero-mass}
Let a unitary Bargmann representation satisfy (G5), or $\hH+\lambda\hM
\geq0$ for some $\lambda$. A vector is invariant under all spatial
translations if and only if it lies in $\Hil_0$. Hence a Galilean net
satisfies (G6$'$) if and only if $\Hil_0=\C\Omega$.
\end{proposition}

\begin{proof}
Invariant vectors have zero mass by Lemma~2 of
Ref.~\onlinecite{Pachon2026a}, which uses only (B1)--(B3). Conversely,
$\hP=0$ on $\ker\hM$ by Lemma~4 there.
\end{proof}

Uniqueness of the vector invariant under space-time translations, which
is (G6), does not suffice.

\begin{proposition}[(G6) does not imply (G6$'$)]\label{prop:counterexample}
There is a Galilean net satisfying (G3)--(G6), with $\dim\Hil_0=2$, for
which $\A(\Op)\Omega\not\subset\C\Omega$ for every non-empty open
$\Op\subset\R^3\times\R$.
\end{proposition}

\begin{proof}
Let $\Hil=\Fock(L^2(\R^3))\otimes\C^2$, with the second-quantised
translations, rotations and boosts of mass $m_0$ tensored with $\one$,
$W=W_{\rm F}\otimes\one$, $\hM=m_0\hN_{\rm F}\otimes\one$ and
\[
\hH:=\rmd\Gamma(h_0)\otimes\one+\hN_{\rm F}\otimes Y_0+\one\otimes
E\ket{2}\bra{2},\qquad Y_0:=E'\one+\lambda\sigma_x,
\]
with $h_0=-\hbar^2\Delta/2m_0$, $E>0$ and $E'\geq\lambda>0$; here
$\{\ket{1},\ket{2}\}$ is the standard basis of $\C^2$ and $\sigma_x$ is
the Pauli matrix. The three summands are non-negative, so $\hH\geq0$.
The last two commute with the translations, rotations and boosts, so
$B(\bv)^{-1}\hH B(\bv)=\hH+\bv\cdot\hP+\frac12|\bv|^2\hM$ as for the free
field, and (G3)--(G5) hold. The vectors invariant under spatial
translations form $\Omega_{\rm F}\otimes\C^2$, on which $\hH=\one\otimes
E\ket{2}\bra{2}$. So $\Omega:=\Omega_{\rm F}\otimes\ket{1}$ is the only
vector invariant under space and time translations, (G6) holds, and
$\Hil_0=\Omega_{\rm F}\otimes\C^2$.

Since $\sigma(\hM)\subset m_0\Z$, the gauge average $X\mapsto\frac{m_0}
{2\pi\hbar}\int_0^{2\pi\hbar/m_0}\rme^{\rmi\theta\hM/\hbar}X
\rme^{-\rmi\theta\hM/\hbar}\rmd\theta$ maps $\F(\Op)$ onto $\A(\Op)$ and
sends $X\Omega$ to its component in $\Hil_0$. It therefore suffices to
find $X\in\F(\Op)$ with $\braket{\Omega_{\rm F}\otimes\ket{2},X\Omega}
\neq0$. Every non-empty open $\Op$ contains slices $B\times\{s\}$ and
$B\times\{t\}$ with $0<|s-t|$ as small as we please. Take $X=\evol{s}W(\epsilon f)\evolm{(s-t)}W(\epsilon g)\evolm{t}$
with $f,g\in L^2(B)$. On the one-particle space $\hH$ acts as $h_0\otimes
\one+\one\otimes Y$ with $Y:=Y_0+E\ket{2}\bra{2}$. Expanding the Weyl
operators, the only contribution of order $\epsilon^2$ to the matrix
element comes from creating a particle with $W(\epsilon g)$ and
annihilating it with $W(\epsilon f)$, and
\[
\braket{\Omega_{\rm F}\otimes\ket{2},X\Omega}=-\frac{\hbar\epsilon^2}{2}\,
\rme^{\rmi sE/\hbar}\,\braket{f,\rme^{-\rmi(s-t)h_0/\hbar}g}\,
\braket{2|\rme^{-\rmi(s-t)Y/\hbar}|1}+O(\epsilon^4).
\]
For $f=g\neq0$ and $0<|s-t|$ small, the first bracket is close to
$\|g\|^2$ and the second equals $-\rmi(s-t)\lambda/\hbar+O((s-t)^2)\neq0$.
\end{proof}

The extra zero-mass state is an internal excitation that the time-zero
fields do not create but the dynamics couples to; the example is artificial, but it shows that (G6$'$), which
is condition~(8) of Dadashev's spectrality axiom~\cite{Dadashev1985}, or the irreducibility of
the time-zero fields in Theorem~\ref{thm:rigidity}(v), has to be
assumed.

\begin{theorem}[Vacuum sector]\label{thm:vacuum}
Let $(\Hil,U,W,\Omega)$ be a Galilean net.
\begin{enumerate}[label=(\alph*)]
\item $\A(\Op)\Omega\subset\Hil_0$ for every $\Op$, and $\Hil_0$ is the
space of vectors invariant under spatial translations.
\item Under (G6$'$), $\A(\Op)\Omega=\C\Omega$ for every $\Op$: the
representation of the observable net on the closure of
$\bigcup_\Op\A(\Op)\Omega$ is one-dimensional.
\item For every non-empty bounded open $B$ and $t\in\R$ there is a
projection $Q\neq0$ in $\A(B,t)$ with $Q\Omega=0$. Hence $\Omega$ is
separating neither for $\A(\Op)$ nor for $\F(\Op)$ whenever $\Op$
contains a slice.
\end{enumerate}
\end{theorem}

\begin{proof}
(a) Elements of $\A(\Op)$ commute with $\hM$, and $\hM\Omega=0$ by
Theorem~\ref{thm:rigidity}(iv); the second statement is
Proposition~\ref{prop:zero-mass}. (b) follows from (a). (c) With $V$ as
in Theorem~\ref{thm:rigidity}, the decomposition $L^2(\R^3)=L^2(B)\oplus
L^2(\R^3\setminus B)$ gives $\Fock(L^2(\R^3))=\Fock(L^2(B))\otimes
\Fock(L^2(\R^3\setminus B))$, and $V\F(B,0)V^{-1}=\B(\Fock(L^2(B)))
\otimes\one\otimes\one$ because the Fock representation of the Weyl
operators over $L^2(B)$ is irreducible. Let $Q_0:=V^{-1}(E\otimes\one
\otimes\one)V$, where $E$ is the projection of $\Fock(L^2(B))$ onto the
states with at least one particle. It lies in $\F(B,0)$, commutes with
$\hM$ by Theorem~\ref{thm:rigidity}(ii), and annihilates $\Omega$ by
Theorem~\ref{thm:rigidity}(iv). Put $Q:=\evol{t}Q_0\evolm{t}$.
\end{proof}

Part~(c) for $\F(\Op)$ is Corollary~2 of Ref.~\onlinecite{Pachon2026a}
and, for regions with temporal extent, Bain's observation~\cite{Bain2011}.
In particular no Galilean net has the Reeh--Schlieder property for its
vacuum, for the field or for the observable algebras, and the vacuum
defines no Tomita--Takesaki modular objects for any local algebra.
Part~(b) is Buchholz's remark for the resolvent algebra of
non-relativistic Bose fields~\cite{Buchholz2020}, derived here from
positivity and covariance. The Reeh--Schlieder question therefore has
content only in sectors of non-zero mass and in states of non-zero
density, which we now treat.

% =====================================================================
\section{Finite mass and time tubes in the vacuum}\label{sec:mass}
% =====================================================================

In this section the time-zero fields are irreducible, so that by
Theorem~\ref{thm:rigidity}(v) $\Hil=\Fock(L^2(\R^3))$, $W=W_{\rm F}$,
$\Omega=\Omega_{\rm F}$ and $\hM=m_0\hN_{\rm F}$. The sector of mass
$nm_0$ is the $n$-particle space $\Hil_n$. For a non-empty bounded open
$B\subset\R^3$ we write $\Fock=\Fock(L^2(B))\otimes\Fock(L^2(B^c))$, with
$B^c:=\R^3\setminus B$, and $\Hil_n=\bigoplus_{k=0}^n\Hil_k^B\otimes
\Hil_{n-k}^{B^c}$ for the corresponding decomposition by the number $k$
of particles in $B$, the eigenvalue of the local number operator $\hN_B:=
\rmd\Gamma(\one_B)$. Theorem~\ref{thm:tube} below also covers Fermi gases
and gases in external potentials, on the corresponding Fock spaces.

\subsection{Finite-mass sectors}

\begin{theorem}[Finite-mass sectors]\label{thm:mass}
Let the time-zero fields of a Galilean net be irreducible, and let
$n\geq1$, $t\in\R$ and $B$ be as above.
\begin{enumerate}[label=(\roman*)]
\item $\A(B,t)\restriction\Hil_n$ is unitarily equivalent to
$\bigoplus_{k=0}^n\B(\Hil_k^B)\otimes\one_{\Hil_{n-k}^{B^c}}$, a type I
von Neumann algebra whose centre is generated by the spectral
projections of the local mass $m_0\evol{t}\hN_B\evolm{t}$.
\item No vector of $\Hil_n$ is cyclic, and none is separating, for
$\A(B,t)\restriction\Hil_n$, nor for $\F(B,t)$.
\item If $\F(\Op)=\B(\Hil)$, then $\A(\Op)\restriction\Hil_n=\B(\Hil_n)$:
every non-zero vector of $\Hil_n$ is cyclic, and none is separating.
\end{enumerate}
\end{theorem}

\begin{proof}
By conjugation with $\evol{t}$, which commutes with $\hM$, we may take
$t=0$. (i) $\F(B,0)=\B(\Fock(L^2(B)))\otimes\one$, and its elements
commuting with $\hN_{\rm F}=\hN_B\otimes\one+\one\otimes\hN_{B^c}$ are
those commuting with $\hN_B\otimes\one$, that is,
$\bigoplus_k\B(\Hil_k^B)\otimes\one$. Restricted to $\Hil_n$ this is the
stated algebra; its centre is spanned by the projections onto the
summands. (ii) Let $P_k$ be these projections and $\Psi\in\Hil_n$. Since
$\Hil_0^B=\C$, the summand $k=0$ is $\C\otimes\Hil_n^{B^c}$, on which
$\A(B,0)$ acts by scalars; so $P_0\A(B,0)\Psi=\C P_0\Psi$, which is not
dense in the infinite-dimensional space $\Hil_n^{B^c}$. Likewise the
commutant $\bigoplus_k\one\otimes\B(\Hil_{n-k}^{B^c})$ acts by scalars on
the summand $k=n$, which is $\Hil_n^B\otimes\C$, so $\Psi$ is not cyclic
for the commutant, that is, not separating. For $\F(B,0)=\B(\Fock(L^2(B)))
\otimes\one$ the argument is the same: the components of $\Psi$ in
$\Fock(L^2(B))$ and in $\Fock(L^2(B^c))$ have at most $n$ particles, so
$\Psi$ is cyclic neither for $\F(B,0)$ nor for its commutant
$\one\otimes\B(\Fock(L^2(B^c)))$. (iii) The elements of
$\B(\Hil)$ commuting with $\hN_{\rm F}$ are $\bigoplus_n\B(\Hil_n)$.
\end{proof}

The statement about $\F(B,t)$ in (ii) holds without irreducibility. By
Theorem~\ref{thm:rigidity}, $\Hil\cong\Fock(L^2(\R^3))\otimes\K$ with
$\hM=m_0\hN_{\rm F}\otimes\one+\one\otimes\hM_\K$ and $\hM_\K\geq0$,
so a vector in the spectral subspace $\hM\leq nm_0$ has at most $n$
quanta in the first factor, and the argument above applies with
$\Fock(L^2(B^c))$ replaced by $\Fock(L^2(B^c))\otimes\K$. The same holds
for $\A(B,t)$: such a vector is annihilated by the non-zero projection of
$\A(B,t)$ onto the states with more than $n$ quanta in $B$ at time $t$,
and $\A(B,t)$ maps it into the spectral subspace $\hM\leq nm_0$, which is
a proper closed subspace.

In a sector of finite mass the fixed-time observable algebras thus carry
a superselection rule for the local mass, which has no relativistic
counterpart, and the Reeh--Schlieder property fails for them in both of
its halves. For regions with temporal extent it fails only in its
separating half, and in a trivial way, provided the tube algebra is
$\B(\Hil)$. We now prove this for interacting Bose and Fermi gases, also in
external potentials.

\subsection{Time tubes}

The algebra of a \emph{time tube} $B\times I$, with $I$ an open interval,
is $\F(B\times I)=\bigvee_{t\in I}\F(B,t)$ (Fig.~\ref{fig:tubes}). Both this section and the
next use the following induction, which is the Reeh--Schlieder argument
with the spectrum condition replaced by analyticity in a strip.

\begin{lemma}[Tube induction]\label{lem:induction}
Let $\M$ be a von Neumann algebra on $\Hil$, $\Omega\in\Hil$ and $L$ a
self-adjoint operator with $L\Omega=0$ and $\rme^{\rmi tL}\M
\rme^{-\rmi tL}=\M$. Suppose that for some $a\in(0,\infty]$ and every
$C\in\M$ the map $t\mapsto\rme^{\rmi tL}C\Omega$ extends to a function
analytic in $\{0<\mathrm{Im}\,z<a\}$ and continuous on its closure. Let
$\Rg_0\subset\M$ be a von Neumann algebra, $\Rg_t:=\rme^{\rmi tL}\Rg_0
\rme^{-\rmi tL}$, and $\Rg_I:=\bigvee_{t\in I}\Rg_t$. If $\Omega$ is
cyclic for $\bigvee_{t\in\R}\Rg_t$, then it is cyclic for $\Rg_I$ for
every non-empty open interval $I$.
\end{lemma}

\begin{proof}
Let $\Upsilon\perp\Rg_I\Omega$. For $A_1,\dots,A_n\in\Rg_0$ put
$G(s_1,\dots,s_n):=\braket{\Upsilon,A_1(s_1)\cdots A_n(s_n)\Omega}$, with
$A(s):=\rme^{\rmi sL}A\rme^{-\rmi sL}$. We show by descending induction on
$k$ that $G=0$ whenever $s_1,\dots,s_k\in I$ and $s_{k+1},\dots,s_n\in\R$.
For $k=n$ this is the hypothesis on $\Upsilon$. Assume it for $k$, fix
$s_1,\dots,s_{k-1}\in I$ and $s_k,\dots,s_n\in\R$, and choose $\sigma_0$
with $s_k-\sigma_0\in I$. Since $L\Omega=0$,
\[
F(\sigma):=G(s_1,\dots,s_{k-1},s_k-\sigma_0+\sigma,\dots,s_n-\sigma_0+
\sigma)=\braket{Y,\rme^{\rmi\sigma L}C\Omega},
\]
with $C:=A_k(s_k-\sigma_0)\cdots A_n(s_n-\sigma_0)\in\M$ and $Y:=
(A_1(s_1)\cdots A_{k-1}(s_{k-1}))^*\Upsilon$. By assumption $F$ is the
boundary value of a function analytic in the strip and continuous on its
closure, and by the induction hypothesis it vanishes for $\sigma$ in a
neighbourhood of $0$. By the Schwarz reflection principle it vanishes
identically, and at $\sigma=\sigma_0$ this is the claim for $k-1$. For
$k=0$ we get $\Upsilon\perp A_1(s_1)\cdots A_n(s_n)\Omega$ for all times and
all $A_j\in\Rg_0$. These vectors span a dense subspace of the closure of
$(\bigvee_{t\in\R}\Rg_t)\Omega$, because the $*$-algebra they generate is
strongly dense in $\bigvee_{t\in\R}\Rg_t$. Hence $\Upsilon=0$.
\end{proof}

The hypothesis on $t\mapsto\rme^{\rmi tL}C\Omega$ holds with $\M=
\B(\Hil)$ and $a=\infty$ when $L\geq0$, and with $a=\beta/2$ for the
Liouvillean of a KMS state (Sec.~\ref{sec:thermal}). For the gases below
it reduces the tube algebra to a unique-continuation property of the
$n$-body dynamics. Without external potential this property follows
from Requardt's analysis of $n$-body Schr\"odinger
operators~\cite{Requardt1986}, as explained after
Theorem~\ref{thm:tube}; we prove it by unique continuation, which needs
no translation invariance.

\paragraph{Gases in external potentials.} The tube theorem holds beyond
the Galilean Bose gases of Proposition~\ref{prop:pair}. We use the
following class of one-body and pair potentials.
\begin{enumerate}[label=(P)]
\item $U,V:\R^3\to\R$ are measurable, $V$ is even, and there is a finite
set $F\subset\R^3$ such that $U=U_+-U_-$ and $V=V_+-V_-$, where
$U_\pm,V_\pm\geq0$, $U_-,V_-\in L^2(\R^3)+L^\infty(\R^3)$, $U_\pm$ are
bounded on every compact subset of $\R^3\setminus F$, and $V_\pm$ on
every compact subset of $\R^3\setminus\{0\}$.
\end{enumerate}
Potentials satisfying~(V), with $U=0$, satisfy~(P). So do harmonic
traps, periodic and other locally bounded potentials that are bounded
below, the Coulomb potentials of fixed charges placed at the points of
$F$, and pair potentials that are repulsive and arbitrarily singular at
the origin, such as $|\bx|^{-12}$, or confining, such as $|\bx|^2$. Let
$\Hil$ be the bosonic Fock space $\Fock(L^2(\R^3))$ or the fermionic one
$\Gamma_{\rm a}(L^2(\R^3))$, with $n$-particle subspaces $\Hil_n$, vacuum
$\Omega=\Omega_{\rm F}$ and number operator $\hN_{\rm F}$. On $\Hil_n$
let $\hH_{U,V}$ act as the self-adjoint operator of the quadratic form
of
\begin{equation}\label{eq:H-UV}
\sum_{j=1}^n\Bigl(-\frac{\hbar^2}{2m_0}\Delta_j+U(\bx_j)\Bigr)+
\sum_{i<j}V(\bx_i-\bx_j),
\end{equation}
which is closed and bounded below (proof of Lemma~\ref{lem:UC}); let
$-E(n):=\inf\sigma(\hH_{U,V}\restriction\Hil_n)$, $\hH:=\hH_{U,V}+
E(\hN_{\rm F})$ and $\hM:=m_0\hN_{\rm F}$. For $U=0$ and $V$
satisfying~(V), with $V_\pm:=\max(\pm V,0)$, this is the Hamiltonian of
Proposition~\ref{prop:pair}: then $V_\pm\in L^2+L^\infty$, the form
domain is $H^1$, and the operator of Proposition~\ref{prop:pair}, being
self-adjoint and contained in the operator of the form, equals it. In
all cases $\hH\geq0$, $\hH\Omega=0$ and $\hH$ commutes with
$\hN_{\rm F}$. The slice algebras are $\F(B,0):=\{W_{\rm F}(g):g\in
L^2(B)\}''$ for bosons and $\F(B,0):=\{a(g),a^*(g):g\in L^2(B)\}''$ for
fermions, $\F(B,t):=\evol{t}\F(B,0)\evolm{t}$, and $\F(\Op)$ and
$\A(\Op)$ are defined by~\eqref{eq:regions} and~\eqref{eq:observables}.
Unless the gas is bosonic and $U=0$, this is not a Galilean net in the
sense of Definition~\ref{def:net}: for non-constant $U$ it is not covariant
under spatial translations and boosts, and for fermions the fields of
disjoint regions anticommute, so that (G2) holds only for the observable
algebras. The lemma and the theorem below use only the structure just
described.

\begin{lemma}[Unique continuation]\label{lem:UC}
Let $(U,V)$ satisfy~(P), let $\hH$ be as above, for bosons or fermions,
and let $B\subset\R^3$ be non-empty and open. Let $n\geq1$,
$\Psi\in\Hil_n$, and suppose that for all $t$ in a non-empty open interval
the $n$-particle wave function $\evolm{t}\Psi$ vanishes almost everywhere
on $U_B:=\{(\bx_1,\dots,\bx_n):\ \bx_j\in B\text{ for some }j\}$. Then
$\Psi=0$.
\end{lemma}

\begin{proof}
Let $\Sigma\subset\R^{3n}$ be the set of configurations with
$\bx_i=\bx_j$ for some $i<j$ or $\bx_j\in F$ for some $j$, and
$\Omega_0:=\R^{3n}\setminus\Sigma$. As a finite union of affine subspaces
of codimension $3$, $\Sigma$ is a closed null set, and $\Omega_0$ is open
and connected. Put $W:=\frac{2m_0}{\hbar^2}\bigl(\sum_jU(\bx_j)+
\sum_{i<j}V(\bx_i-\bx_j)\bigr)$, which by~(P) is bounded on every
compact subset of $\Omega_0$, and split it as $W=W_+-W_-$ according to
the signs in~(P). The form $q_+(u):=\|\nabla u\|^2+\int W_+|u|^2$, on the
functions $u\in H^1(\R^{3n})$ with $\int W_+|u|^2<\infty$, is closed, as
the sum of two closed non-negative forms, and densely defined, since its
domain contains $C_c^\infty(\Omega_0)$, which is dense because $\Sigma$
is a null set. By Kato's theorem~\cite{Kato1951}, as in the proof of
Proposition~\ref{prop:pair}, multiplication by $W_-$ has relative bound
zero with respect to $-\Delta$ on $L^2(\R^{3n})$, hence relative form
bound zero with respect to $q_+$, so by the KLMN theorem
$q:=q_+-\int W_-|\cdot|^2$ is closed and bounded below on the domain of
$q_+$. Its self-adjoint operator $k$ commutes with permutations of the
particles, and on symmetric, respectively antisymmetric, functions
$\hH\restriction\Hil_n=\frac{\hbar^2}{2m_0}k+E(n)$; so
$\evolm{t}\Psi=\rme^{-\rmi tE(n)/\hbar}\rme^{-\rmi(\hbar t/2m_0)k}\Psi$
for $\Psi\in\Hil_n$, and the hypothesis says that $\rme^{-\rmi sk}\Psi$
vanishes almost everywhere on $U_B$ for all $s$ in a non-empty open
interval. If $u\in D(k)$ and $\phi\in C_c^\infty(\Omega_0)$, then
$q(\phi,u)=\braket{\phi,ku}$, that is, $-\Delta u+Wu=ku$ in the sense of
distributions on $\Omega_0$. Since $Wu\in L^2_{\rm loc}(\Omega_0)$, also
$\Delta u\in L^2_{\rm loc}(\Omega_0)$, and $u\in H^2_{\rm loc}(\Omega_0)$
by interior elliptic regularity. Masuda's theorem~\cite{Masuda1967}, in
the form of Lemma~A7 of Ref.~\onlinecite{Pachon2026a} applied to $k$ and
the open set $U_B$, gives $\Psi=0$.
\end{proof}

\begin{theorem}[Time tubes in the vacuum]\label{thm:tube}
Let $(U,V)$ satisfy~(P), and consider the Bose or Fermi gas with the
Hamiltonian $\hH$ above; for instance the Bose gases of
Proposition~\ref{prop:pair}, or particles in a harmonic trap or in the
Coulomb field of fixed charges. For every non-empty bounded open $B$ and
every non-empty open interval $I$,
\[
\F(B\times I)=\B(\Hil).
\]
Hence $\Omega$ is cyclic but not separating for $\F(B\times I)$, and
$\A(B\times I)\restriction\Hil_n=\B(\Hil_n)$ for every $n$. If
$\hH_{U,V}+\lambda\hM\geq0$ for some $\lambda\in\R$, the same holds for
the net of the physical dynamics $\hH_{U,V}$.
\end{theorem}

\begin{proof}
(a) \emph{The vacuum projection.} Under the decomposition
$\Hil=\Gamma(L^2(B))\otimes\Gamma(L^2(\R^3\setminus B))$, with $\Gamma$
the bosonic or fermionic Fock space, given for fermions by the unitary of
Ref.~\onlinecite{Derezinski2006}, Theorems~55(3) and~56,
$\F(B,0)=\B(\Gamma(L^2(B)))\otimes\one$. Let $E\in\F(B,0)$ be the
projection onto the states with no particle in $B$, $E_t:=\evol{t}E\evolm{t}$, and
$P:=\bigwedge_{t\in I}E_t\in\F(B\times I)$. The range of $P$ consists of
the $\Psi$ for which $\evolm{t}\Psi$ has no particle in $B$ for every
$t\in I$. Since $\hH$ commutes with $\hN_{\rm F}$, this holds for each
$n$-particle component separately, and Lemma~\ref{lem:UC} gives
$\Psi\in\C\Omega$. As $\Omega$ is in the range, $P=\ket{\Omega}
\bra{\Omega}$.

(b) \emph{Cyclicity.} Let $\Hil_{\rm c}$ be the closure of
$\mathfrak A\Omega$, where $\mathfrak A$ is the $*$-algebra generated by
$\bigcup_{t\in\R}\F(B,t)$. It is invariant under $\mathfrak A''$, under
$\evol{t}$ and under the gauge group, so its projection commutes with
$\hN_{\rm F}$, with $\evol{t}$ and with the fields $\psio(g)$, $\psiod(g)$,
$g\in L^2(B)$; for fermions $\psio(g)$ and $\psiod(g)$ stand here and
below for $a(g)$ and $a^*(g)$. We show $\Hil_n\subset\Hil_{\rm c}$ by induction on $n$,
the case $n=0$ being clear. Since the projection onto $\Hil_{\rm c}$
commutes with $\hN_{\rm F}$, it suffices to show that a vector
$\Psi\in\Hil_n$ orthogonal to $\Hil_{\rm c}$ vanishes. For $g\in L^2(B)$, $t\in\R$ and $\xi\in\Hil_{n-1}\subset
\Hil_{\rm c}$ the vector $\evol{t}\psiod(g)\xi$ lies in $\Hil_{\rm c}$,
so $\braket{\psio(g)\evolm{t}\Psi,\xi}=0$. As $\psio(g)\evolm{t}\Psi\in
\Hil_{n-1}$, it vanishes for all $g\in L^2(B)$, that is, $\evolm{t}\Psi$
has no particle in $B$, for every $t$. By Lemma~\ref{lem:UC}, $\Psi=0$.
Hence $\Omega$ is cyclic for $\bigvee_{t\in\R}\F(B,t)$, and by
Lemma~\ref{lem:induction} with $L=\hH/\hbar\geq0$ and $\M=\B(\Hil)$ it is
cyclic for $\F(B\times I)$.

(c) $\F(B\times I)$ contains $X\ket{\Omega}\bra{\Omega}Y^*=
\ket{X\Omega}\bra{Y\Omega}$ for $X,Y\in\F(B\times I)$. By (b) these
approximate every rank-one operator in norm, so $\F(B\times I)$ contains
all compact operators and equals $\B(\Hil)$. The statement about
$\A(B\times I)$ follows as in the proof of Theorem~\ref{thm:mass}(iii):
the elements of $\B(\Hil)$ commuting with $\hN_{\rm F}$ are
$\bigoplus_n\B(\Hil_n)$.

(d) If $\hH_{U,V}+\lambda\hM\geq0$, the net of $\hH_{U,V}$ is that of
$\hH_{U,V}+\lambda\hM$, and (a)--(c) apply with $\hH_{U,V}+\lambda\hM$ in
place of $\hH$: Lemma~\ref{lem:UC} involves the $n$-particle Hamiltonians only
up to additive constants, and Lemma~\ref{lem:induction} only needs a
non-negative generator.
\end{proof}

For bosons with $V=0$, with $U=0$ or with the Newtonian potentials of
static stars and of the Schwarzschild exterior, the theorem is Proposition~3 of
Ref.~\onlinecite{Pachon2026c},
where the algebra is generated by time-smeared fields.
Theorem~\ref{thm:tube} is a Galilean timelike tube theorem
(Fig.~\ref{fig:tubes}). In
relativistic theories the algebra of a thin tube around a timelike
segment contains that of the double cone spanned by the
segment~\cite{Borchers1961,StrohmaierWitten2023}, and with the
time-slice property, as for the free field, the algebra of a time tube
is that of its causal completion; in Galilean spacetime the causal completion of any region with temporal extent is the
whole spacetime, and the theorem says that its algebra is the algebra of
all bounded operators. Beyond the structure described before Lemma~\ref{lem:UC}, the proof uses
the potentials only through that lemma. For $U=0$ and pair potentials
whose relative bound with respect to the kinetic energy is less than one,
a class that contains~(V), Lemma~\ref{lem:UC} is contained in Requardt's
work~\cite{Requardt1986}. His Theorem~6 concerns nested evolutions: for $t_1,\dots,
t_n\in I$, the function $g_1$ is evolved with the one-particle
Hamiltonian over the time $t_1-t_2$, multiplied by $g_2$ and evolved
with the two-particle Hamiltonian over $t_2-t_3$, and so on, the last
evolution being with the $n$-particle Hamiltonian over $t_n$. The
theorem states that the resulting functions are total in
$L^2(\R^{3n})$ when $g_n(\bx_n)\cdots g_1(\bx_1)$ ranges over the
products supported in a given open subset of $\R^{3n}$. For supports in
$B^n$ the last step is $\rme^{-\rmi t_n\hH/\hbar}\Phi$ with $\Phi$
vanishing outside the set of Lemma~\ref{lem:UC}, which gives the lemma
after the time interval is reflected and translated; the constants
$E(n)$ contribute phases, and on $\Hil_n$ one takes symmetric or
antisymmetric parts.
Since $\hH\Omega=0$, the symmetrised, respectively antisymmetrised,
functions are, up to phases, the
products $\evol{s_n}\psiod(g_n)\evolm{s_n}\cdots\evol{s_1}\psiod(g_1)
\evolm{s_1}\Omega$ with $s_k=-t_k$, and the theorem applied to $-I$
gives step~(b) directly, without Lemma~\ref{lem:induction}. Requardt
proves it by analytic continuation in the times and positions of the
individual particles, using only that the joint spectrum of energy and
momentum of each subsystem lies above a paraboloid. This bound needs
translation invariance. His theorem applies to antisymmetric as well as
symmetric wave functions, and the bound also holds for non-negative pair
potentials of any strength, so for $U=0$ the singular repulsive potentials of~(P) are
within reach of his argument. His density theorem for a single
Schr\"odinger operator (his Theorem~2) assumes the unique continuation
property and a generalised eigenfunction expansion, which he takes from
Agmon's theory of short-range potentials; for $n\geq2$ the $n$-body
potential is not short range on $\R^{3n}$, and that theorem does not
reach the $n$-body problem. The proof given here uses instead Masuda's
unique continuation theorem for the Schr\"odinger
evolution~\cite{Masuda1967}, in the form of Lemma~A7 of
Ref.~\onlinecite{Pachon2026a}, which passes to the elliptic operator
$\partial_y^2+\Delta-W$ in one more variable. It needs neither
translation invariance nor decay of the potentials, so it covers traps,
locally bounded periodic potentials and the Coulomb fields of finitely
many fixed charges; it needs
local boundedness away from the collision set and from $F$, which
Requardt's argument does not.

\begin{figure}[htb]
\centering
\begin{tikzpicture}[font=\small, >={Stealth[length=2mm]},
  tube/.style={draw=blue!70!black, fill=blue!25, line width=0.6pt},
  ttl/.style={font=\small\bfseries, align=center}]
\def\figW{2.6}\def\figH{2.1}\def\figR{0.22}\def\figT{0.95}
\begin{scope}
  \fill[blue!7] (-\figT,0) -- (0,\figT) -- (\figT,0) -- (0,-\figT) -- cycle;
  \draw[blue!70!black, dashed] (-\figT,0) -- (0,\figT) -- (\figT,0) -- (0,-\figT) -- cycle;
  \draw[->, gray] (-\figW,0) -- (\figW,0) node[below left=-1pt, black] {$\bx$};
  \draw[->, gray] (0,-\figH) -- (0,\figH) node[left, black] {$ct$};
  \draw[tube] (-\figR,-\figT) rectangle (\figR,\figT);
  \draw[gray!70!black] (0.55,0.45) -- (1.2,1.1)
    node[above right=-2pt, font=\footnotesize, blue!70!black] {double cone};
  \draw[gray!70!black] (\figR,-0.6) -- (0.95,-1.25)
    node[below right=-2pt, font=\footnotesize, blue!70!black] {$B\times I$};
  \node[ttl] at (0,\figH+0.45) {(a) Minkowski};
\end{scope}
\begin{scope}[xshift=6.6cm]
  \fill[blue!7] (-\figW+0.1,-\figH+0.1) rectangle (\figW-0.1,\figH-0.1);
  \draw[->, gray] (-\figW,0) -- (\figW,0) node[below left=-1pt, black] {$\bx$};
  \draw[->, gray] (0,-\figH) -- (0,\figH) node[left, black] {$t$};
  \draw[tube] (-0.7,-0.55) rectangle (0.7,0.55);
  \node[font=\footnotesize, blue!70!black] at (0,-0.16) {$B\times I$};
  \draw[->, blue!70!black, line width=0.7pt] (0.35,\figH-0.3) -- (0.35,0.62);
  \node[right, font=\footnotesize] at (0.35,\figH-0.45) {Lemma~\ref{lem:induction}};
  \draw[->, blue!70!black, line width=0.7pt] (0.35,-\figH+0.3) -- (0.35,-0.62);
  \draw[->, blue!70!black, line width=0.7pt] (0.75,0.3) -- (\figW-0.4,0.3)
    node[above left, font=\footnotesize, black] {Lemma~\ref{lem:UC}};
  \draw[->, blue!70!black, line width=0.7pt] (-0.75,0.3) -- (-\figW+0.4,0.3);
  \node[font=\footnotesize, gray!60!black, align=center] at (-1.45,1.35)
    {causal completion:\\ all of spacetime};
  \node[font=\footnotesize, align=center] at (-1.45,-1.3)
    {$\F(B\times I)=\B(\Hil)$\\ in the vacuum};
  \node[ttl] at (0,\figH+0.45) {(b) Galilean};
\end{scope}
\end{tikzpicture}
\caption{The algebra of a time tube, in one space dimension (schematic;
in (a) light rays are at $45^\circ$, in (b) the aspect ratio is
arbitrary). (a) In Minkowski space the algebra of a thin tube around a
timelike segment contains that of the double cone spanned by the segment
(timelike tube theorem~\cite{Borchers1961,StrohmaierWitten2023}).
(b) In Galilean spacetime the causal completion of every region with
temporal extent is the whole spacetime, and in the vacuum of the free
field and of the gases of Theorem~\ref{thm:tube} the algebra
of $B\times I$ is $\B(\Hil)$. In the proof, the
projection onto the states with no particle in $B$ at any time in $I$
lies in $\F(B\times I)$ and, by Lemma~\ref{lem:UC}, is the vacuum
projection (horizontal arrows: Lemma~\ref{lem:UC} reaches all of space
from $B$). The fields of $B$ at times in $I$ generate the Fock space from
the vacuum, by Lemma~\ref{lem:UC} at all times and
Lemma~\ref{lem:induction} (vertical arrows: Lemma~\ref{lem:induction}
brings cyclicity back from all times to $I$), or, for $U=0$, directly by Requardt's
density theorem and its proof~\cite{Requardt1986}. Hence $\F(B\times I)$ contains every
rank-one operator.}
\label{fig:tubes}
\end{figure}

\subsection{Locality at bounded energy}\label{sec:bounded-energy}

A time tube generates the algebra of all bounded operators, but only
through analytic continuation in time (Lemma~\ref{lem:induction}), and
by Proposition~4 of Ref.~\onlinecite{Pachon2026a} the norm of the vacuum
commutator between the field algebras of two regions at different times
does not decrease with their distance. Both effects involve test functions of arbitrarily
high momenta: compactly supported functions have unbounded momentum
support, and the norm in Proposition~4 of Ref.~\onlinecite{Pachon2026a}
is approached by test functions boosted to the velocity $d/|t|$. At
bounded energy the one-particle propagation is local up to errors that
decay faster than any power of the distance beyond the maximal
velocity; this is the free case of the maximal-velocity bounds of
Schr\"odinger dynamics~\cite{ArbunichPusateriSigalSoffer2021}. The cost
of reaching a distant region from a time tube then lies in the local
pieces, not in the norm of the local operator.

\begin{proposition}[Locality at bounded energy]\label{prop:bounded-energy}
Let $h_0:=-\hbar^2\Delta/2m_0$ on $L^2(\R^3)$, let $\chi\in
C_c^\infty(\R)$ be real with $E_\chi:=\max\supp\chi>0$, and put
$v_\chi:=(2E_\chi/m_0)^{1/2}$.
\begin{enumerate}[label=(\alph*)]
\item For all $N\in\N$, $T>0$ and $\delta>0$ there is $C$ such that
\[
\bigl\|\one_{B_1}\chi(h_0)^2\rme^{\rmi th_0/\hbar}\one_{B_2}\bigr\|\leq
C\,(d-v_\chi|t|)^{-N}
\]
for all measurable $B_1,B_2\subset\R^3$ at distance $d$ and all
$|t|\leq T$ with $d-v_\chi|t|\geq\delta$.
\item Let $B\subset\R^3$ be bounded and open, let $w\in L^2(\R^3)$ vanish
outside a set at distance $d$ from $B$, and let $u=\sum_j\rme^{\rmi
t_jh_0/\hbar}g_j$ be a finite sum with $g_j\in L^2(B)$ and $|t_j|\leq
T$. If $\|\chi(h_0)(u-w)\|\leq\epsilon\|\chi(h_0)w\|$ with $\epsilon<1$
and $d-v_\chi T\geq\delta$, then, with $C$ as in~(a),
\[
\sum_j\|g_j\|\geq\frac{1-\epsilon}{C}\,
\frac{\|\chi(h_0)w\|^2}{\|w\|}\,(d-v_\chi T)^N.
\]
\item For the free Bose field in the vacuum, let $B\times I$ be a time
tube, $w\in L^2(\R^3)$ and $0<\sigma\leq1$. There is $X$ in the
$*$-algebra generated by the Weyl operators $\evol{t}W(g)\evolm{t}$,
$t\in I$, $g\in L^2(B)$, such that
\[
\|X\|\leq2\|w\|/\sigma,\qquad \|X\Omega-w\|\leq3\sigma^2\|w\|,
\]
where $w$ also denotes the one-particle vector $a^*(w)\Omega$; the bounds
do not depend on the distance between $B$ and the support of $w$.
\end{enumerate}
\end{proposition}

\begin{proof}
(a) The operator $\chi(h_0)^2\rme^{\rmi th_0/\hbar}$ is the convolution
with
\[
K_t(\mathbf r)=(2\pi)^{-3}\int\psi(\bk)\,\rme^{\rmi\Phi(\bk)}\,\rmd^3k,
\qquad\psi(\bk):=\chi\bigl(\hbar^2|\bk|^2/2m_0\bigr)^2,\quad
\Phi(\bk):=\frac{\hbar t|\bk|^2}{2m_0}+\bk\cdot\mathbf r .
\]
On $\supp\psi$, $\hbar|\bk|/m_0\leq v_\chi$, so $|\nabla\Phi|\geq
D:=|\mathbf r|-v_\chi|t|$. For $D>0$ we integrate by parts $N+3$ times
with $L:=-\rmi|\nabla\Phi|^{-2}\nabla\Phi\cdot\nabla$, which satisfies
$L\rme^{\rmi\Phi}=\rme^{\rmi\Phi}$. The second derivatives of $\Phi$ are
$\hbar t/m_0$ times the identity, so each term of $(L^{\mathsf T})^{N+3}
\psi$ is bounded by a constant times $|\partial^\gamma\psi|\,D^{-N-3}
(\hbar|t|/m_0D)^{N+3-|\gamma|}$, and $|K_t(\mathbf r)|\leq C'D^{-N-3}$
for $D\geq\delta$, uniformly in $|t|\leq T$. By Schur's test the norm in
(a) is at most $\int_{|\mathbf r|\geq d}|K_t(\mathbf r)|\,\rmd^3r\leq
4\pi C'\int_{d-v_\chi|t|}^\infty(s+v_\chi T)^2s^{-N-3}\,\rmd s$, which is
bounded by a constant times $(d-v_\chi|t|)^{-N}$ when $d-v_\chi|t|\geq
\delta$.

(b) By~(a), applied to each term with $B_1$ the set outside which $w$
vanishes and $B_2=B$,
\[
(1-\epsilon)\|\chi(h_0)w\|^2\leq|\braket{\chi(h_0)w,\chi(h_0)u}|\leq
C\,\|w\|\,(d-v_\chi T)^{-N}\sum_j\|g_j\|,
\]
where the first inequality holds because $|\braket{\chi(h_0)w,\chi(h_0)
u}|\geq\|\chi(h_0)w\|^2-\|\chi(h_0)w\|\,\|\chi(h_0)(u-w)\|$.

(c) For $w=0$ take $X=0$. Otherwise, by Lemma~\ref{lem:density} there
is a finite sum $u=\sum_j\rme^{\rmi t_jh_0/\hbar}g_j$, $t_j\in I$,
$g_j\in L^2(B)$, with $\|u-w\|<\sigma^2\|w\|$, so that $u\neq0$. With $c:=(\hbar/2)^{1/2}$, $s:=\sigma/
(c\|u\|)$ and $X:=(W(su)-W(-su))/(2\rmi cs)$, the Weyl
relations~\eqref{eq:weyl} make $W(\pm su)$ a phase times the product of
the $\evol{t_j}W(\pm sg_j)\evolm{t_j}$, so $X$ lies in the stated
$*$-algebra, and $\|X\|\leq1/(cs)=\|u\|/\sigma$. Since $W(v)\Omega=
\rme^{-c^2\|v\|^2/2}\rme^{\rmi c\,a^*(v)}\Omega$,
\[
X\Omega=\rme^{-\sigma^2/2}\sum_{n\ \mathrm{odd}}\frac{(\rmi cs)^{n-1}}
{n!}\,a^*(u)^n\Omega,
\]
whose one-particle component is $\rme^{-\sigma^2/2}u$ and whose other
components have norm at most $\sigma^2\|u\|/2$. Hence $\|X\Omega-u\|
\leq\sigma^2\|u\|$, and the bounds follow from $\|u\|\leq(1+\sigma^2)
\|w\|$.
\end{proof}

By~(b), a time tube reproduces a particle localised beyond the reach
$v_\chi T$ of the maximal velocity, at energies in $\supp\chi$, only
through smearing functions whose total norm grows faster than any power
of $d-v_\chi T$. By~(c), for the free field these large local pieces
cancel almost exactly: the local operator itself can be chosen with norm
of order $\|w\|$ at every distance, each of its Weyl factors creating
many particles and their product only a few. The non-locality of the
time tubes is therefore not paid for in the norm of local operators but
in the fine-tuned cancellation of strong local operations. For the free
Bose field the commutators of $\psio(\chi(h_0)f)$ and $\evol{t}
\psiod(\chi(h_0)g)\evolm{t}$ are the c-numbers $\hbar\braket{\chi(h_0)f,
\rme^{\rmi th_0/\hbar}\chi(h_0)g}$, and for the free Fermi field the
anticommutators of $a(\chi(h_0)f)$ and $\tau_t(a^*(\chi(h_0)g))$ are
$\braket{\chi(h_0)f,\rme^{\rmi th_0/\hbar}\chi(h_0)g}$; so~(a) bounds them
in every state, including the thermal states~(B) of Sec.~\ref{sec:free},
where the grand-canonical dynamics adds only a phase (in an external
potential, as in Proposition~\ref{prop:external}, one-particle
propagation is $\rme^{\rmi th/\hbar}$, which~(a) does not control); for
the Bose and Fermi gases of Theorem~\ref{thm:tube} with $U=0$ it bounds
their vacuum expectation values. The filter $\chi(h_0)$ is not invariant under
boosts, which shift momenta; this is why~(a) is compatible with
Proposition~4 of Ref.~\onlinecite{Pachon2026a}, whose supremum is
approached by test functions of velocity $d/|t|$, outside the range
of $\chi$ as soon as $d>v_\chi|t|$.

% =====================================================================
\section{Thermal states}\label{sec:thermal}
% =====================================================================

Thermal states are not invariant under Galilean boosts~\cite{Sewell2008},
and nothing in this section uses Bargmann covariance. We formulate the
result for a von Neumann algebra $\M$ with a dynamics $\tau$, that is, a
$\sigma$-weakly continuous one-parameter group of automorphisms, and a
state $\omega=\braket{\Omega,\cdot\,\Omega}$ with $\Omega$ cyclic for
$\M$ that satisfies the KMS condition with parameter $0<\beta<\infty$
relative to $\tau$; for $\tau_t=\mathrm{Ad}\,\rme^{\rmi t\hH/\hbar}$
and inverse temperature $\beta_{\rm th}$ the parameter is $\beta=\hbar
\beta_{\rm th}$. Then $\Omega$ is separating for
$\M$~\cite{HaagHugenholtzWinnink1967}, and $\tau_t=\mathrm{Ad}\,
\rme^{\rmi tL}$ for a unique self-adjoint $L$ with $L\Omega=0$, the
Liouvillean. The modular operator is $\Delta=\rme^{-\beta L}$, since the
modular group of $\omega$ is $\tau_{-\beta s}$~\cite{Takesaki1970}.

Let $\Rg(B)\subset\M$ be von Neumann algebras labelled by the non-empty
bounded open sets $B\subset\R^3$, the fixed-time algebras of the
theory, and put $\Rg(B,t):=\tau_t(\Rg(B))$ and $\Rg(B\times I):=
\bigvee_{t\in I}\Rg(B,t)$ for an open interval $I$.

\begin{theorem}[Thermal Reeh--Schlieder property]\label{thm:thermalRS}
In this setting, let $B\subset\R^3$ be non-empty, bounded and open. The
following are equivalent:
\begin{enumerate}[label=(\alph*)]
\item \emph{tube completeness},
\begin{equation}\label{eq:tube-completeness}
\bigvee_{t\in\R}\Rg(B,t)=\M ;
\end{equation}
\item $\Omega$ is cyclic for $\Rg(B\times I)$ for every non-empty open
interval $I$;
\item $\Omega$ is cyclic for $\Rg(B\times I)$ for some non-empty open
interval $I$.
\end{enumerate}
In any case $\Omega$ is separating for $\Rg(B\times I)$. If~(a) holds
and, moreover, a compact group $G$ acts on $\M$ by automorphisms
$\alpha_\gamma$, with $\gamma\mapsto\alpha_\gamma(X)$ $\sigma$-weakly
continuous for every $X\in\M$, that commute with $\tau$, leave $\omega$
invariant and map $\Rg(B)$ onto itself, then $\Omega$ is cyclic and
separating for the fixed-point algebra $\Rg(B\times I)^G$ acting on the
closure of $\M^G\Omega$.
\end{theorem}

\begin{proof}
$\Omega$ is separating for $\M\supset\Rg(B\times I)$.

(a)$\Rightarrow$(b). For $C\in\M$, $C\Omega\in D(\Delta^{1/2})=
D(\rme^{-\beta L/2})$ by Tomita--Takesaki theory. Since
$\rme^{-2y\lambda}\leq1+\rme^{-\beta\lambda}$ for $\lambda\in\R$ and
$0\leq y\leq\beta/2$, also $C\Omega\in D(\rme^{-yL})$, and
$z\mapsto\rme^{\rmi zL}C\Omega$ is analytic in $\{0<\mathrm{Im}\,z<
\beta/2\}$ and continuous on its closure. Lemma~\ref{lem:induction} with
$a=\beta/2$, $\Rg_0=\Rg(B)$ and~\eqref{eq:tube-completeness} gives
cyclicity.

(b)$\Rightarrow$(c) is clear.

(c)$\Rightarrow$(a). The von Neumann algebra $\mathcal N:=\bigvee_{t\in\R}
\Rg(B,t)$ contains $\Rg(B\times I)$, so $\Omega$ is cyclic for
$\mathcal N$. It is invariant under $\tau$, hence under the modular group
$\sigma^\omega_s=\tau_{-\beta s}$, and by Takesaki's
theorem~\cite{Takesaki1972} there is a normal conditional expectation
$E$ of $\M$ onto $\mathcal N$ with $\omega\circ E=\omega$. For $X\in\M$
and $Y\in\mathcal N$,
$\braket{Y\Omega,X\Omega}=\omega(Y^*X)=\omega(E(Y^*X))=\omega(Y^*E(X))=
\braket{Y\Omega,E(X)\Omega}$, so $X\Omega-E(X)\Omega$ is orthogonal to
the dense subspace $\mathcal N\Omega$. Hence $X\Omega=E(X)\Omega$, and
$X=E(X)\in\mathcal N$ because $\Omega$ is separating for $\M$.

For the last statement, $U(\gamma)X\Omega:=\alpha_\gamma(X)\Omega$
defines, by invariance of $\omega$ and cyclicity of $\Omega$, a strongly
continuous unitary representation of $G$ that implements $\alpha_\gamma$
and fixes $\Omega$. Since $\alpha_\gamma$ commutes with $\tau$, it maps
each $\Rg(B,t)$, and hence $\Rg(B\times I)$, onto itself. Let $P_G$ be
the projection onto the invariant vectors. For $X\in\Rg(B\times I)$,
$P_GX\Omega=\int_GU(\gamma)XU(\gamma)^*\rmd\gamma\,\Omega$, with
$\rmd\gamma$ the normalised Haar measure, and the integral lies in
$\Rg(B\times I)^G$. Hence $\Rg(B\times I)^G\Omega$ is dense in
$P_G\Hil$, which contains the closure of $\M^G\Omega$.
\end{proof}

If~\eqref{eq:tube-completeness} holds for every $B$, then $\Omega$ has
the Reeh--Schlieder property of Definition~\ref{def:RS} for the net
$\Op\mapsto\Rg(\Op):=\bigvee\{\Rg(B,t):B\times\{t\}\subset\Op\}$,
because every non-empty open $\Op$ contains a time tube $B\times I$ and
$\Rg(B\times I)\subset\Rg(\Op)\subset\M$. Tube completeness is thus not merely sufficient: it is equivalent to
the Reeh--Schlieder property of the time tubes over $B$. Without it the
algebra $\bigvee_{t\in\R}\Rg(B,t)$ is a proper subalgebra of $\M$
invariant under the modular group, for which $\Omega$ is never cyclic;
this happens, for instance, when some degrees of freedom never enter
$B$. For the free Bose and Fermi
gases~\eqref{eq:tube-completeness} holds by
Theorem~\ref{thm:quasifree}(ii), which proves more: $\Rg(B\times I)=\M$
for every $I$. There the dynamics of the KMS condition is $\tau^\mu$,
with parameter $\beta\hbar$ in the notation of Sec.~\ref{sec:free}, and
its tube algebras are those of the free dynamics $\tau$, because
$\tau^\mu_t=\tau_t\circ\gamma_{-\mu t/\hbar}$ and the gauge group maps
$L^2(B)$ onto itself. It also holds for ideal gases in external
potentials, including traps (Proposition~\ref{prop:external}), and it
survives bounded interactions localised in $B$
(Proposition~\ref{prop:local}). The theorem is the Galilean analogue of
J\"akel's thermal
Reeh--Schlieder theorem~\cite{Jaekel2000}, which derives the property
from locality, additivity and the relativistic KMS condition; here tube
completeness replaces additivity, and the KMS strip of the dynamics
replaces the relativistic KMS condition. The same use of the KMS strip
appears in Strohmaier's proof for quasi-free states on stationary
spacetimes~\cite{Strohmaier2000}.
For fixed-time algebras the thermal vector is separating, but in
general not cyclic (Theorem~\ref{thm:quasifree}(iv)); the Reeh--Schlieder
property thus needs temporal extent in thermal states as well.

In the vacuum representation the same argument, with $L=\hH/\hbar\geq0$
and analyticity in the upper half-plane, gives cyclicity for time tubes
whenever tube completeness holds (step (b) of the proof of
Theorem~\ref{thm:tube}), while separation fails
(Theorem~\ref{thm:vacuum}); this is the Galilean form of Requardt's
theorem~\cite{Requardt1982}. The contrast between the two cases is the
contrast between a one-sided spectrum and a two-sided one: a KMS
Liouvillean has spectrum on both sides of zero, and its strip of
analyticity is attached to a separating vector.

\paragraph{Bounded local interactions.} Tube completeness survives
interactions localised in $B$. In the setting of
Theorem~\ref{thm:thermalRS}, for a self-adjoint $V\in\M$ let
$\tau^V_t(X):=\rme^{\rmi t(L+V)}X\rme^{-\rmi t(L+V)}$ be the perturbed
dynamics, and put $\Rg^V(B,t):=\tau^V_t(\Rg(B))$ and $\Rg^V(B\times J):=
\bigvee_{t\in J}\Rg^V(B,t)$ for an interval $J$. Since $L$ has the
dimension of an inverse time, so has $V$, and $\hbar V$ is the
perturbation of the Hamiltonian that generates $\tau$.

\begin{proposition}[Bounded local interactions]\label{prop:local}
In the setting of Theorem~\ref{thm:thermalRS}, let $V=V^*\in\Rg(B)$.
\begin{enumerate}[label=(\alph*)]
\item $\tau^V$ is a $\sigma$-weakly continuous one-parameter group of
automorphisms of $\M$, and $\bigvee_{t\in\R}\Rg^V(B,t)=\bigvee_{t\in\R}
\Rg(B,t)$. If $\Rg(B\times I)=\M$ for every non-empty open interval $I$,
then $\Rg^V(B\times I)=\M$ for every such $I$.
\item There is a faithful normal state $\omega^V$ on $\M$ that satisfies
the KMS condition with parameter $\beta$ relative to
$\tau^V$~\cite{Araki1973}. In its GNS representation
$(\pi_V,\Hil_V,\Omega_V)$ the vector $\Omega_V$ is separating for
$\pi_V(\M)$, and it is cyclic for $\pi_V(\Rg^V(B\times I))$ for some,
equivalently every, non-empty open interval $I$ if and only if $\Omega$
is cyclic for $\Rg(B\times I)$, that is, if and only if $B$ satisfies
tube completeness for $\tau$.
\end{enumerate}
\end{proposition}

\begin{proof}
(a) Let $W_t:=\rme^{\rmi tL}\rme^{-\rmi t(L+V)}$, a strongly continuous
family of unitaries. On the domain of $L$,
$\frac{\rmd}{\rmd t}W_t=-\rmi\,\tau_t(V)W_t=-\rmi\,W_t\,\tau^V_t(V)$, with
$\tau_t(V)=\rme^{\rmi tL}V\rme^{-\rmi tL}$ and $\tau^V_t(V)=
\rme^{\rmi t(L+V)}V\rme^{-\rmi t(L+V)}$ strongly continuous and of norm
$\|V\|$. Iterating the corresponding integral equations gives, for
$t>0$, the norm-convergent Dyson series
\[
W_t=\sum_{n\geq0}(-\rmi)^n\!\!\int\limits_{0<s_1<\dots<s_n<t}\!\!
\tau_{s_n}(V)\cdots\tau_{s_1}(V)\,\rmd^ns
=\sum_{n\geq0}(-\rmi)^n\!\!\int\limits_{0<s_1<\dots<s_n<t}\!\!
\tau^V_{s_1}(V)\cdots\tau^V_{s_n}(V)\,\rmd^ns,
\]
with the integrals taken in the strong operator topology, and similarly
for $t<0$. So $W_t$ lies in every von Neumann algebra that contains
$\tau_s(V)$ for all $s$ between $0$ and $t$, and in every one that
contains $\tau^V_s(V)$ for these $s$. In particular $W_t\in\M$, and since
$\rme^{\rmi t(L+V)}=W_t^*\rme^{\rmi tL}$,
\begin{equation}\label{eq:cocycle}
\tau^V_t(X)=W_t^*\,\tau_t(X)\,W_t,\qquad X\in\M ;
\end{equation}
thus $\tau^V$ is a group of automorphisms of $\M$, $\sigma$-weakly
continuous because it is implemented by a strongly continuous unitary
group. Let $\mathcal N:=\bigvee_{t\in\R}\Rg(B,t)$ and $\mathcal N^V:=
\bigvee_{t\in\R}\Rg^V(B,t)$. Both contain $V$, and they are invariant
under $\tau$, respectively $\tau^V$; so $W_t\in\mathcal N$ by the first
series and $W_t\in\mathcal N^V$ by the second. By~\eqref{eq:cocycle},
$\mathcal N$ is then invariant under $\tau^V$ and contains $\Rg(B)$, so
$\mathcal N^V\subset\mathcal N$; and $\mathcal N^V$ is invariant under
$\tau_t=\mathrm{Ad}\,W_t\circ\tau^V_t$, so $\mathcal N\subset\mathcal
N^V$.

For the second statement let $\epsilon>0$ and $0<t<\epsilon$. By the
second series $W_t\in\Rg^V(B\times[0,\epsilon))$, so
$\tau_t(X)=W_t\tau^V_t(X)W_t^*\in\Rg^V(B\times[0,\epsilon))$ for
$X\in\Rg(B)$, and $\M=\Rg(B\times(0,\epsilon))\subset\Rg^V(B\times
[0,\epsilon))$. Every non-empty open interval contains some $(a,b)$ with
$a<b$ finite, and $\Rg^V(B\times J)$ increases with $J$; so let
$I=(a,b)$. For $a<a'<b$, applying $\tau^V_{a'}$ with $\epsilon=b-a'$
gives $\Rg^V(B\times[a',b))=\M$, and $\Rg^V(B\times I)$ contains this
algebra.

(b) Let $\Delta=\rme^{-\beta L}$ be the modular operator of $\Omega$.
Araki's perturbation theory~\cite{Araki1973} gives, for every
self-adjoint $k\in\M$, a faithful normal state $\omega_k$ on $\M$ whose
modular group is $X\mapsto\rme^{\rmi s(\log\Delta+k)}X\rme^{-\rmi s(\log
\Delta+k)}$. For $k=-\beta V$ this is $\tau^V_{-\beta s}$, so
$\omega^V:=\omega_{-\beta V}$ is a KMS state with parameter $\beta$
relative to $\tau^V$~\cite{Takesaki1970}. As $\omega^V$ is faithful and
normal, $\pi_V$ is a normal isomorphism of $\M$ onto the von Neumann
algebra $\pi_V(\M)$, which maps generated von Neumann algebras onto
generated von Neumann algebras, and $\Omega_V$ is cyclic and separating
for $\pi_V(\M)$. Theorem~\ref{thm:thermalRS} therefore applies to
$(\pi_V(\M),\pi_V\circ\tau^V\circ\pi_V^{-1},\braket{\Omega_V,\cdot\,
\Omega_V})$ with the algebras $\pi_V(\Rg(B))$, and to $(\M,\tau,\omega)$;
by~(a), both cyclicity conditions are equivalent to tube completeness of
$B$ for $\tau$.
\end{proof}

For the thermal states of the free Bose and Fermi gases, and for those
of the ideal gases of Proposition~\ref{prop:external}, $\Rg(B\times
I)=\M$ for every $B$ and every $I$
(Theorem~\ref{thm:quasifree}(ii), Proposition~\ref{prop:external}(b)),
and $\Rg(K)\subset\Rg(B)$ for $K\subset B$ by~\eqref{eq:MV}. There the
dynamics of the KMS condition is the grand-canonical one, $\tau^\mu$, and
it is $\tau^\mu$ that Proposition~\ref{prop:local} perturbs. For a
self-adjoint $V\in\Rg(K)$, with $K$ non-empty, bounded and open, the
proposition gives $\Rg^V(B\times I)=\M$ for every $B\supset K$ and every
$I$, and the thermal Reeh--Schlieder property of the time tubes over
these $B$ in the KMS state of the perturbed dynamics. Examples are
$V=\sum_{ijkl}v_{ijkl}\,a^*(f_i)a^*(f_j)a(f_k)a(f_l)$ with $f_1,\dots,
f_N\in L^2(K)$ and $v_{ijkl}=\overline{v_{lkji}}$, a two-body interaction
restricted to finitely many local modes, for fermions, and, for bosons,
$V=F(N_f)$ with $F$ real and bounded and $N_f$ the occupation number of a
normalised mode $f\in L^2(K)$: by the Stone--von Neumann theorem the
Weyl operators $W(zf)$, $z\in\C$, generate a multiple of the
Schr\"odinger representation, so $N_f:=\hbar^{-1}\psiod(f)\psio(f)$ is
self-adjoint and affiliated with $\Rg(K)$. Both are gauge invariant. For
gauge-invariant $V$ the tube algebras of the perturbation of $\tau^\mu$
are those of the perturbation of $\tau$ by $V$, because the gauge group
commutes with $\tau$ and with $V$ and maps $\Rg(B)$ onto itself; and
$\omega^V$ is gauge invariant, because the unitaries that implement the
gauge group and fix $\Omega$ commute with $\Delta$ and with $V$, so the
last statement of Theorem~\ref{thm:thermalRS} applies as well. These dynamics are interacting, but they are not
Galilean covariant, and the argument does not reach the tubes over
regions that do not contain $K$. For the ideal Fermi gas on the lattice
$\Z^\nu$, $\nu\geq3$, with a small local two-body perturbation, Botvich
and Malyshev proved that the M\o ller morphisms, which intertwine the
free and the perturbed dynamics, are invertible, so that the two
Liouvilleans are unitarily equivalent~\cite{BotvichMalyshev1983}; since
these isomorphisms do not preserve localisation, whether a continuum
analogue reaches those tubes is open.

% =====================================================================
\section{The free field and ideal gases}\label{sec:free}
% =====================================================================

\subsection{Quasi-free states}

Let $\mathfrak h:=L^2(\R^3)$ and $h_0:=-\hbar^2\Delta/2m_0$. We treat
bosons and fermions together. For bosons the field algebra is generated
by the Weyl operators~\eqref{eq:weyl}; for fermions by the bounded
operators $a(g)$, antilinear in $g$, with $\{a(f),a^*(g)\}=\braket{f,g}$.
In both cases the free dynamics and the gauge group act by $\tau_t$:
$g\mapsto\rme^{\rmi th_0/\hbar}g$ and $\gamma_\theta$: $g\mapsto
\rme^{\rmi\theta}g$ on the test functions, and the grand-canonical
dynamics at chemical potential $\mu$ is $\tau^\mu_t$: $g\mapsto
\rme^{\rmi t(h_0-\mu)/\hbar}g$. The fixed-time algebra $\F(B,t)$ is now
the $C^*$-algebra generated by the $\tau_t(W(g))$, respectively the
$\tau_t(a(g))$, with $g\in L^2(B)$, and $\F$ the one generated by all
$g\in\mathfrak h$; their von Neumann closures in a representation are
defined below.

A gauge-invariant quasi-free state $\omega$ is fixed by its one-particle
density $\rho$, a bounded positive operator ($\rho\leq\one$ for
fermions), through $\omega(W(g))=\exp(-\frac\hbar4\braket{g,(\one+2\rho)
g})$ for bosons and $\omega(a^*(f)a(g))=\braket{g,\rho f}$ with vanishing
truncated higher functions for fermions. We take $\rho=\nu(h_0)$ with
$\nu\geq0$ bounded and finite density
\begin{equation}\label{eq:density}
n:=(2\pi)^{-3}\int\nu\bigl(\hbar^2|\bk|^2/2m_0\bigr)\,\rmd^3k<\infty.
\end{equation}
Since $\omega(\psiod(g)\psio(g))=\hbar\braket{g,\rho g}$ for bosons,
$\rho$ is the one-particle density of $\hbar^{-1/2}\psio$ for bosons and
of $a$ for fermions, and $n$ is the particle density. We consider three cases:
\begin{enumerate}[label=(\Alph*)]
\item the vacuum, $\nu=0$;
\item thermal states, $\nu(\varepsilon)=(\rme^{\beta(\varepsilon-\mu)}
\mp1)^{-1}$ with $0<\beta<\infty$, and $\mu<0$ for bosons (upper sign)
and $\mu\in\R$ for fermions (lower sign); these are the KMS states of
$\tau^\mu$ at inverse temperature $\beta$, whose KMS parameter in the
time variable of $\tau^\mu$ is $\beta\hbar$;
\item the Fermi sea, $\nu=\one_{[0,\varepsilon_{\rm F})}$ with
$\varepsilon_{\rm F}>0$, so that $\rho=P$ is the projection onto the
Fermi ball $|\bk|<k_{\rm F}:=(2m_0\varepsilon_{\rm F})^{1/2}/\hbar$.
\end{enumerate}

Their GNS representations can be realised on Fock spaces over a larger
one-particle space~\cite{ArakiWoods1963,Derezinski2006}. Let
$\K:=\mathfrak h$ and $Tg:=g$ in case (A),
$\K:=\mathfrak h\oplus\overline{\mathfrak h}$ in case (B), and
$\K:=P^\perp\mathfrak h\oplus\overline{P\mathfrak h}$ in case (C), and
in cases (B) and (C) define the real-linear map $T:\mathfrak h\to\K$ by
\begin{equation}\label{eq:T}
Tg:=\bigl((\one\pm\rho)^{1/2}g,\ \overline{\rho^{1/2}g}\bigr)
\quad\text{(bosons: $+$; fermions: $-$)},
\end{equation}
which in case (C) reads $Tg=(P^\perp g,\overline{Pg})$. Then
$\pi(W(g))=W_{\rm F}(Tg)$ for bosons and $\pi(a(g)+a^*(g))=a_{\rm F}(Tg)+
a_{\rm F}^*(Tg)$ for fermions, on the Fock space over $\K$ with $\Omega$
its vacuum, is the GNS representation in the three cases; in (B)
and~(C), $\Omega$ is cyclic by Theorem~\ref{thm:quasifree}(iv)
and~(v). These are the representations of Araki
and Woods~\cite{ArakiWoods1963} for bosons and of Araki and Wyss for
fermions, in the form of Ref.~\onlinecite{Derezinski2006}, Theorems~40
and~43; case (C) is a Fock representation relative to the Fermi sea. For a closed real
subspace $\mathcal V\subset\K$ let $\M(\mathcal V)$ be the von Neumann
algebra generated by the Weyl operators, or the fermionic fields
$a_{\rm F}(v)+a_{\rm F}^*(v)$, with $v\in\mathcal V$. Then
\begin{equation}\label{eq:MV}
\Rg(B,0)=\M\bigl(T L^2(B)\bigr),\qquad
\Rg(B\times I)=\M\Bigl(\overline{\textstyle\sum_{t\in I}T\rme^{\rmi th_0/
\hbar}L^2(B)}\Bigr),\qquad \M=\M(T\mathfrak h),
\end{equation}
by the lattice properties of $\M(\cdot)$ (Ref.~\onlinecite{Derezinski2006},
Theorems~54(4) and~55(5)) and because $a(g)$ is a linear combination of fields
$a(g')+a^*(g')$. Here $\Rg(B,t):=\pi(\F(B,t))''$, $\Rg(B\times I):=
\bigvee_{t\in I}\Rg(B,t)$ and $\M:=\pi(\F)''$.

\subsection{Classification}

\begin{theorem}[Reeh--Schlieder dichotomy for the free
field]\label{thm:quasifree}
In each of the cases (A)--(C), for every non-empty bounded open
$B\subset\R^3$, $t\in\R$ and non-empty open interval $I$:
\begin{enumerate}[label=(\roman*)]
\item $\Rg(B,t)$ is a type I$_\infty$ factor;
\item $\Rg(B\times I)=\M$;
\item in case (A), $\Omega$ is neither cyclic nor separating for
$\Rg(B,t)$, and $\M=\B(\Hil)$, so that $\Omega$ is cyclic but not
separating for $\Rg(B\times I)$;
\item in case (B), $\Omega$ is separating but not cyclic for
$\Rg(B,t)$; $\M$ is a factor of type III$_1$, for which $\Omega$ is
cyclic and separating, and the modular group of $\omega$ is the
grand-canonical dynamics, $\sigma^\omega_s=\tau^\mu_{-\beta\hbar s}$;
\item in case (C), $\Omega$ is separating but not cyclic for $\Rg(B,t)$,
and $\M=\B(\Hil)$, so that $\Omega$ is cyclic but not separating for
$\Rg(B\times I)$.
\end{enumerate}
\end{theorem}

\begin{table}[htb]
\caption{The Reeh--Schlieder property of the free Galilean field
(Theorem~\ref{thm:quasifree}): type of the algebra, and whether the GNS
vector is cyclic (c) or separating (s), for a spatial region $B$ at fixed
time and for a time tube $B\times I$. A dash means that the property
fails.}
\label{tab:classification}
\begin{ruledtabular}
\begin{tabular}{llll}
State & Fixed time & Time tube & Tube algebra\\
\hline
Vacuum (A) & I$_\infty$; --, -- & c, -- & $\B(\Hil)$, type I$_\infty$\\
Thermal (B) & I$_\infty$; --, s & c, s & global factor, type III$_1$\\
Fermi sea (C) & I$_\infty$; --, s & c, -- & $\B(\Hil)$, type I$_\infty$\\
\end{tabular}
\end{ruledtabular}
\end{table}

\begin{proof}
The unitary implementing $\tau_t$ in the GNS representation leaves
$\Omega$ invariant, so it suffices to take $t=0$. We use the criteria of
Lemma~\ref{lem:real-subspaces}: $\Omega$ is cyclic for $\M(\mathcal V)$
if and only if $\mathcal V+\rmi\mathcal V$ is dense in $\K$, and
separating if and only if $\mathcal V\cap\rmi\mathcal V=\{0\}$.

(ii) $\|Tg\|^2=\braket{g,(\one+2\rho)g}\geq\|g\|^2$ for bosons and
$\|Tg\|^2=\|g\|^2$ for fermions, so $T$ has closed range and
$\overline{T\mathcal D}=T\mathfrak h$ for every dense subspace
$\mathcal D$. The complex span $\mathcal D$ of the $\rme^{\rmi th_0/
\hbar}L^2(B)$, $t\in I$, is dense by Lemma~\ref{lem:density}, and the
real span of their images is $T\mathcal D$, since each
$\rme^{\rmi th_0/\hbar}L^2(B)$ is a complex subspace. By~\eqref{eq:MV},
$\Rg(B\times I)=\M(T\mathfrak h)=\M$.

(iii) Here $\Hil=\Fock(\mathfrak h)$ and $\Rg(B,0)=\M(L^2(B))=
\B(\Fock(L^2(B)))\otimes\one$, a type I$_\infty$ factor; $L^2(B)$ is a
complex subspace which is neither dense nor $\{0\}$; and $\M=\B(\Hil)$.

(iv) Write $\mathcal V_B:=TL^2(B)$. For $g_1,g_2\in L^2(B)$,
$Tg_1+\rmi Tg_2=((\one\pm\rho)^{1/2}(g_1+\rmi g_2),\overline{\rho^{1/2}
(g_1-\rmi g_2)})$, so $\mathcal V_B+\rmi\mathcal V_B=(\one\pm\rho)^{1/2}
L^2(B)\oplus\overline{\rho^{1/2}L^2(B)}$. The operator $(\one\pm\rho)^{1/2}$
is bounded with bounded inverse, since $\one-\rho\geq(1+\rme^{\beta\mu})
^{-1}$ for fermions, so the first summand is a closed proper subspace of
$\mathfrak h$ and $\Omega$ is not cyclic. If $Tg=\rmi Tg'$ with $g,g'\in
L^2(B)$, then $(\one\pm\rho)^{1/2}(g-\rmi g')=0$ and $\rho^{1/2}(g+\rmi
g')=0$. Both operators are injective, because $\nu>0$, so $g=g'=0$ and
$\Omega$ is separating. For the type, the restriction of $\omega$ to the
fields over $L^2(B)$ is the gauge-invariant quasi-free state with density
$\rho_B:=\one_B\rho\one_B$ on $L^2(B)$. The operator $\rho^{1/2}\one_B$
has the kernel $\kappa(\bx-\by)\one_B(\by)$, where $\kappa\in
L^2(\R^3)$ is the inverse Fourier transform of $\nu(\hbar^2|\bk|^2/2m_0)
^{1/2}$ and $\|\kappa\|_2^2=n$ by Plancherel; so $\|\rho^{1/2}\one_B\|
_{\rm HS}^2=n|B|$, and $\rho_B$ is trace class with $\Tr\rho_B=n|B|$.
By Theorem~5.1 of Ref.~\onlinecite{Chaiken1968} for bosons, and by the
Powers--St\o rmer criterion~\cite{PowersStormer1970} for fermions, the
GNS representation of this restriction is quasi-equivalent to the Fock
representation over $L^2(B)$, so its von Neumann algebra is a type
I$_\infty$ factor. Because $\Omega$ is separating for $\Rg(B,0)$, the
reduction of $\Rg(B,0)$ to the closure of $\Rg(B,0)\Omega$ is an
isomorphism onto that algebra. The statements about $\M$ are
Theorems~40(5), (6), (9) and~43(5), (6), (9) of
Ref.~\onlinecite{Derezinski2006} with $\gamma=\rme^{-\beta(h_0-\mu)}$,
which has purely continuous spectrum and trivial kernel and, for bosons,
satisfies $\gamma\leq\rme^{\beta\mu}<1$: $\M$ is a type III$_1$ factor, $\Omega$ is cyclic and
separating, and the modular operator is $\Gamma(\gamma\oplus
\overline\gamma^{-1})$, which implements $g\mapsto\gamma^{\rmi s}g=
\rme^{-\rmi s\beta(h_0-\mu)}g$, that is, $\tau^\mu_{-\beta\hbar s}$.

(v) Now $\mathcal V_B+\rmi\mathcal V_B=P^\perp L^2(B)\oplus\overline{P
L^2(B)}$ by the same computation. By Lemma~\ref{lem:gap} there is a
non-zero $u\in\mathfrak h$ with $u=0$ on $B$ and $Pu=0$. Then $u\in
P^\perp\mathfrak h$ is orthogonal to $P^\perp L^2(B)$, and $\Omega$ is
not cyclic. If $Tg=\rmi Tg'$ with $g,g'\in L^2(B)$, then $g-\rmi g'\in
L^2(B)\cap P\mathfrak h$ and $g+\rmi g'\in L^2(B)\cap P^\perp\mathfrak
h$. A function in $L^2(B)$ has a real-analytic Fourier transform, which
cannot vanish on the Fermi ball or on its complement unless it vanishes
identically; so both intersections are $\{0\}$, $g=g'=0$, and $\Omega$ is
separating. The type follows as in (iv), with $\Tr\rho_B=n|B|$. Finally
$T$ maps $\mathfrak h$ onto $\K$, since $g=u+v$ with $u\in P^\perp
\mathfrak h$ and $v\in P\mathfrak h$ has $Tg=(u,\overline v)$; hence
$\M=\M(\K)=\B(\Hil)$.
\end{proof}

Thermal states are thus cyclic and separating for the algebras of time
tubes, but not for fixed-time algebras. More generally, for
$\rho=\nu(h_0)$ satisfying~\eqref{eq:density}, $\Omega$ is separating
for $\Rg(B,t)$ if and only if $n>0$: by the computation in (iv),
separation fails if and only if some non-zero $f\in L^2(B)$ has
$\nu\hat f=0$ or, for fermions, $(1-\nu)\hat f=0$ almost everywhere,
and the Fourier transform of $f$ is entire. The separating property is
an effect of non-zero density, not of temperature, as the Fermi sea
shows. Cyclicity at fixed time fails in all three cases. In (A) and (B)
the particle component of $\mathcal V_B+\rmi\mathcal V_B$ is the image
of $L^2(B)$ under the invertible operator $(\one\pm\rho)^{1/2}$; in (C) a
function supported outside $B$ can have its Fourier transform supported
outside the Fermi ball (Lemma~\ref{lem:gap}). In both cases nothing
plays the role of the anti-locality of $(-\Delta+m^2)^{1/2}$ that drives
relativistic cyclicity.

\subsection{Ideal gases in external potentials}

Let $U$ satisfy the conditions on $U$ in~(P), and let $h:=h_0+U$ be the
form sum, that is, $\hbar^2/2m_0$ times the operator $k$ of the proof of
Lemma~\ref{lem:UC} with $n=1$; it is self-adjoint and bounded below. Let
$\tau_t$ act on the test functions by $g\mapsto\rme^{\rmi th/\hbar}g$ and
$\tau^\mu_t$ by $g\mapsto\rme^{\rmi t(h-\mu)/\hbar}g$. For a bounded
Borel function $\nu\geq0$ on $\sigma(h)$, with $\nu\leq1$ for fermions,
let $\omega$ be the gauge-invariant quasi-free state with one-particle
density $\rho:=\nu(h)$, represented on the Fock space over $\mathfrak
h\oplus\overline{\mathfrak h}$ by~\eqref{eq:T}, as in case~(B) above with
$h$ in place of $h_0$ (for general $\nu$ this need not be the GNS
representation), and let $\Rg(B,t)$, $\Rg(B\times I)$ and $\M$ be defined
as there.

\begin{proposition}[Ideal gases in external potentials]
\label{prop:external}
In this setting, for every non-empty bounded open $B\subset\R^3$ and
every non-empty open interval $I$:
\begin{enumerate}[label=(\alph*)]
\item the vectors $\rme^{\rmi th/\hbar}g$ with $t\in I$ and $g\in L^2(B)$
span a dense subspace of $L^2(\R^3)$;
\item $\Rg(B\times I)=\M$; in particular tube
completeness~\eqref{eq:tube-completeness} holds for every $B$;
\item if $\nu(\varepsilon)=(\rme^{\beta(\varepsilon-\mu)}\mp1)^{-1}$
with $0<\beta<\infty$, and $\mu<\inf\sigma(h)$ for bosons, then $\omega$
is a KMS state of $\tau^\mu$ with parameter $\beta\hbar$, and $\Omega$ is
cyclic and separating for $\M$, hence for $\Rg(B\times I)$;
\item if, in addition, $\Tr\rme^{-\beta h}<\infty$, as for a harmonic
trap $U(\bx)=\frac12m_0\omega_0^2|\bx|^2$, then $\omega$ is the Gibbs
state of $\rmd\Gamma(h)-\mu\hN_{\rm F}$ at inverse temperature $\beta$ on
the Fock space over $L^2(\R^3)$, and $\M$ is a factor of type
I$_\infty$;
\item if, under the hypotheses of~(c), $h$ has some continuous spectrum,
$\M$ is a factor of type III$_1$.
\end{enumerate}
\end{proposition}

\begin{proof}
(a) If $w\in L^2(\R^3)$ is orthogonal to these vectors, then
$\rme^{-\rmi th/\hbar}w$ vanishes almost everywhere on $B$ for $t\in I$.
On $\Hil_1$ the Hamiltonian of Sec.~\ref{sec:mass} with this $U$ is
$h+E(1)$, so Lemma~\ref{lem:UC} with $n=1$ gives $w=0$.

(b) Since $\rho$ commutes with $h$, $T\rme^{\rmi th/\hbar}g=(\rme^{\rmi
th/\hbar}\oplus\overline{\rme^{\rmi th/\hbar}})Tg$, and~\eqref{eq:MV}
holds with $\rme^{\rmi th/\hbar}$ in place of $\rme^{\rmi th_0/\hbar}$.
The proof of Theorem~\ref{thm:quasifree}(ii) applies verbatim, with~(a)
in place of Lemma~\ref{lem:density}.

(c) With $\gamma:=\rme^{-\beta(h-\mu)}$, which has trivial kernel and,
for bosons, norm less than one, $\rho=\gamma(\one\mp\gamma)^{-1}$.
Theorems~40(6), (11), (12) and~43(6), (13), (14) of
Ref.~\onlinecite{Derezinski2006}, with $(h-\mu)/\hbar$ as one-particle
generator and $\beta\hbar$ as inverse temperature, and which do not
require continuous spectrum, give that $\Omega$ is cyclic and separating
for $\M$, that the modular operator is $\Gamma(\gamma\oplus\overline
\gamma^{-1})$, so that the modular group of $\omega$ is
$\tau^\mu_{-\beta\hbar s}$, and that $\omega$ is a KMS state of $\tau^\mu$
with parameter $\beta\hbar$. By~(b), $\Rg(B\times I)=\M$; the tube
algebras of $\tau^\mu$ are those of $\tau$, since the gauge group maps
$L^2(B)$ onto itself.

(d) Now $h$ has purely discrete spectrum, and $\gamma$ is trace class,
since $\Tr\gamma=\rme^{\beta\mu}\Tr\rme^{-\beta h}$. So
$\Gamma(\gamma)=\rme^{-\beta(\rmd\Gamma(h)-\mu\hN_{\rm F})}$ is trace
class, with trivial kernel, on the bosonic or fermionic Fock space
$\Gamma(\mathfrak h)$, where $\Gamma(\gamma)$ is the second quantisation of
$\gamma$, and the Gibbs state $X\mapsto\Tr(\Gamma(\gamma)X)/
\Tr\Gamma(\gamma)$ is a faithful normal state on $\B(\Gamma(\mathfrak
h))$. By Theorems~49 and~50 of Ref.~\onlinecite{Derezinski2006}, on the
confined Bose and Fermi gas, the representation used here is unitarily
equivalent to the standard representation of $\B(\Gamma(\mathfrak h))$
on $\Gamma(\mathfrak h)\otimes\overline{\Gamma(\mathfrak h)}$, in which
$\Omega$ corresponds to the vector of this Gibbs state. So $\omega$ is
the Gibbs state, and $\M\cong\B(\Gamma(\mathfrak h))$ is a factor of
type I$_\infty$.

(e) Then $\gamma$ has some continuous spectrum, and Theorems~40(5), (9)
and~43(5), (9) of Ref.~\onlinecite{Derezinski2006} give the type.
\end{proof}

In a harmonic trap, and in every trap with $\Tr\rme^{-\beta h}<\infty$,
the thermal Reeh--Schlieder property of time tubes thus holds with a
global algebra of type I. The type III$_1$ of the thermal states of the
homogeneous gas in Theorem~\ref{thm:quasifree}(iv) comes from the
continuous spectrum of the one-particle Hamiltonian, by~(e); a trap with
$\Tr\rme^{-\beta h}<\infty$ removes it together with the type III$_1$, as
a box does for relativistic Gibbs states, and the thermal
Reeh--Schlieder property of time tubes survives.
The limits of the thermal states of static stars in
Ref.~\onlinecite{Pachon2026c}, Proposition~4, are the states of~(c) with
$U=m_0\Phi$, $\Phi$ the Newtonian potential of the star, and
$\mu<\inf\sigma(h)$; by~(b) and~(c) their time tubes generate the global
algebra, for which the KMS vector is cyclic and separating. That algebra
is of type III$_1$ by~(e) whenever $h$ has some continuous spectrum, for
instance for a spherically symmetric star of finite radius, outside which
$\Phi=-GM/|\bx|$: then $\sigma(h)\supset[0,\infty)$, and $h$ has no
eigenvalue in $(0,\infty)$, because outside the star the radial solutions
of positive energy are Coulomb wave functions, which are not square
integrable.

% =====================================================================
\section{Type, modular flow and crossed products}\label{sec:consequences}
% =====================================================================

\begin{corollary}[The local type depends on the state]\label{cor:type}
Let $\Op\subset\R^3\times\R$ be non-empty and open, and put $\Rg(\Op):=
\bigvee\{\Rg(B,t):B\times\{t\}\subset\Op\}$. For the free Galilean
field in the states (A)--(C), $\Rg(\Op)=\M$: a type I$_\infty$ factor,
$\B(\Hil)$, in the vacuum and in the Fermi sea, and a type III$_1$
factor in the thermal states (B), where the modular group of $\omega$ on
it is $\tau^\mu_{-\beta\hbar s}$. For the gases of
Theorem~\ref{thm:tube} in their vacuum, $\F(\Op)=\B(\Hil)$. In all these
cases the fixed-time field algebras of bounded regions are type
I$_\infty$ factors.
\end{corollary}

\begin{proof}
$\Op$ contains a time tube $B\times I$, so $\Rg(\Op)$ lies between
$\Rg(B\times I)$ and $\M$, and Theorems~\ref{thm:tube}
and~\ref{thm:quasifree} apply. For these gases, $\F(B,t)=\evol{t}
(\B(\Gamma(L^2(B)))\otimes\one)\evolm{t}$, with $\Gamma$ the bosonic or
fermionic Fock space, as in step~(a) of the proof of
Theorem~\ref{thm:tube}.
\end{proof}

By Proposition~\ref{prop:external}(b),(d),(e), also for the thermal
states of ideal gases in external potentials $\Rg(\Op)=\M$, a type
I$_\infty$ factor in a harmonic trap and a type III$_1$ factor when the
one-particle Hamiltonian has some continuous spectrum.

The dependence on the state is that of the global algebra, which is of
type I in pure states and of type III$_1$ in these KMS states, as in
relativistic theories. What is Galilean is that every region with
temporal extent generates it. The fixed-time observable algebras are of
type I but not factors: their centre is generated by the local number
operator, by Theorem~\ref{thm:mass}(i) in the vacuum and, in the states
(B) and (C), because the gauge group acts on $\Rg(B,t)\cong
\B(\Fock(L^2(B)))$ through the unitaries $\rme^{\rmi\theta\hN_B}$ of the
algebra itself.

In relativistic theories the algebras of double cones are of type
III$_1$ in the vacuum under general
assumptions~\cite{BuchholzDAntoniFredenhagen1987}, hence in every state
locally normal to the vacuum; for the free field this includes all
quasi-free Hadamard states~\cite{Verch1997}. For the free Galilean field,
fixed-time algebras are of type I$_\infty$ in all the states considered,
and the algebras of regions with temporal extent take the
state-dependent type of the global algebra. The modular group of a
thermal state on the algebra of a time tube is the grand-canonical
dynamics with time rescaled by $-\beta\hbar$, which is the physical
dynamics on gauge-invariant observables. For the ideal gases, also in external potentials
(Proposition~\ref{prop:external}(b),(c)), the thermal time of Connes and
Rovelli~\cite{ConnesRovelli1994} attached to any
region with temporal extent is therefore the physical time, but only
because that region generates the global algebra: a proper von Neumann
subalgebra for which $\Omega$ is cyclic is never invariant under the
modular group of $\omega$, as in the proof of
Theorem~\ref{thm:thermalRS}~\cite{Takesaki1972}.

Modular crossed products replace type III algebras by semifinite ones
with a trace.

\begin{proposition}[Crossed products and local
entropy]\label{prop:crossed}
In the setting of Theorem~\ref{thm:quasifree}:
\begin{enumerate}[label=(\alph*)]
\item every modular group of a faithful normal state on a fixed-time
algebra $\Rg(B,t)$ is inner, and the crossed product $\Rg(B,t)
\rtimes_\sigma\R$ is isomorphic to $\Rg(B,t)\,\bar\otimes\,L^\infty(\R)$,
an algebra of type I with centre $L^\infty(\R)$; in each of the cases
(A)--(C) the restriction of $\omega$ to $\Rg(B,t)$ has von Neumann
entropy $\sum_js(\lambda_j)<\infty$, with $\lambda_j$ the eigenvalues of
$\rho_B:=\one_B\rho\one_B$ and $s$ the Bose or Fermi entropy function
below; it vanishes in case (A);
\item in case (B) the crossed product $\M\rtimes_{\sigma^\omega}\R$ is a
factor of type II$_\infty$, and by Theorem~\ref{thm:quasifree}(ii) it is
the crossed product of the algebra of every time tube.
\end{enumerate}
\end{proposition}

\begin{proof}
(a) $\Rg(B,t)$ is isomorphic to $\B(\K_B)$ for a Hilbert space $\K_B$. A
faithful normal state on it is $\Tr(D\,\cdot)$ with a density matrix $D$
of trivial kernel, and its modular group is $\mathrm{Ad}\,D^{\rmi s}$.
For an action by $\mathrm{Ad}\,u_s$ with a unitary group $u_s$ in the
algebra, the unitary $(\xi(s))_s\mapsto(u_s\xi(s))_s$ on
$L^2(\R,\K_B)$ carries the crossed product onto the tensor product with
the group von Neumann algebra of $\R$, which is abelian. For the entropy
we may take $t=0$, as in the proof of Theorem~\ref{thm:quasifree}; the
restriction of $\omega$ is then the normal quasi-free state with
one-particle density $\rho_B$. In an eigenbasis of $\rho_B$ it is a
product over the modes of states with mean occupations $\lambda_j$, the
eigenvalues of $\rho_B$, so its entropy is $\sum_js(\lambda_j)$ with
$s(x)=(1+x)\log(1+x)-x\log x$ for bosons and $s(x)=-x\log x-(1-x)\log(1-x)$
for fermions. Choose $\chi\in C_c^\infty(\R^3)$ with $\chi=1$ on $B$. The
kernel of $\rho$ is $\rho(\bx-\by)$ with $\rho(\cdot)$ smooth,
because it is the Fourier transform of an integrable function of compact
support or of exponential decay. So $\chi\rho\chi$ has a smooth kernel of
compact support and lies in every Schatten class, as one sees by
factoring out powers of the harmonic oscillator; hence so does $\rho_B=
\one_B\chi\rho\chi\one_B$. Since $s(x)\leq Cx^{1/2}$ for $0\leq x\leq
\|\rho\|$, with $\|\rho\|\leq1$ for fermions, $\sum_js(\lambda_j)
<\infty$.

(b) By Takesaki's duality theory the crossed product of a type III
factor by its modular group is of type II$_\infty$~\cite{Takesaki1973},
and it is a factor exactly when the algebra is of type
III$_1$~\cite{ConnesTakesaki1977}.
\end{proof}

\begin{corollary}[What is relativistic is the locality of type
III$_1$]\label{cor:locality}
For the free Klein--Gordon field, each double cone has its own type
III$_1$ factor in the vacuum and in quasi-free Hadamard
states~\cite{BuchholzDAntoniFredenhagen1987,Verch1997}, hence its own
type II$_\infty$ crossed product, and no normal state on it has a density
matrix. For the free Galilean field, fixed-time field algebras are type
I factors on which the vacuum, the thermal states and the Fermi sea have
finite entropy; in the states (A)--(C), among the fixed-time field
algebras and the field algebras of open regions, type III$_1$ algebras
and type II$_\infty$ modular crossed products occur only for the global
algebra of the thermal states, with which the field algebra of every
region with temporal extent coincides.
\end{corollary}

\begin{proof}
The Galilean statements are Theorem~\ref{thm:quasifree},
Corollary~\ref{cor:type} and Proposition~\ref{prop:crossed}. The
relativistic ones are the cited results and the argument of the proof of
Proposition~\ref{prop:crossed}(b), applied to the local algebras, on
which the vacuum is faithful by the Reeh--Schlieder theorem; a type III
factor has no normal semifinite trace, so its normal states have no
density matrices.
\end{proof}

Corollary~\ref{cor:locality} sharpens the contrast drawn in the
companion papers. In the non-relativistic limit of the Klein--Gordon
field, in the vacuum, the time-zero algebras pass from type III$_1$ to
type I, the
vacuum loses its entanglement, and thermal states with the chemical
potential shifted by the rest energy survive
(Ref.~\onlinecite{Pachon2026c}, Theorem~4 and Proposition~4). The
present results show what survives: the limiting thermal states are the
bosonic states of case (B), whose fixed-time algebras are type I and whose
algebra of any region with temporal extent is the global type III$_1$
factor. The thermal Reeh--Schlieder property does not need type
III$_1$: in a harmonic trap the global algebra of a thermal state is of
type I, and the KMS vector is cyclic and separating for every time tube
(Proposition~\ref{prop:external}(c),(d)).

% =====================================================================
\section{Discussion}\label{sec:discussion}
% =====================================================================

The Reeh--Schlieder property of relativistic quantum field theory
splits, in Galilean theories, along two lines: the state and the
temporal extent of the region. In the vacuum, positivity and boost
covariance force the time-zero fields into a direct sum of Fock
representations, the vacuum sector of the observables is
one-dimensional once the vacuum is the only spatially invariant vector,
and no local algebra has the vacuum as a separating vector; for the free
field and for interacting Bose and Fermi gases, also in external
potentials, with the renormalised dynamics if the interaction is not
H-stable, the algebras of regions with temporal extent are nevertheless
$\B(\Hil)$.
At bounded energy free one-particle propagation is local up to errors that
decay faster than any power of the distance beyond the maximal
velocity, but for the free field a time tube still reaches a distant
region with a local operator whose norm does not depend on the
distance, through local pieces that cancel almost exactly
(Proposition~\ref{prop:bounded-energy}). In
a KMS state the property holds for a time tube $B\times I$ if and only if
the spatial algebras of $B$ at all times generate the global algebra, and
then for every region that contains such a tube; we prove this for ideal
gases, also in external potentials, and show that it survives bounded
interactions localised in $B$. At fixed time the thermal vector is
always separating, and for the free gases it is not cyclic. For the free
field the type of local algebras is state dependent: fixed-time field
algebras are of type I in all the states considered, and among field
algebras type III$_1$ appears only through the global algebra, which
every region with temporal extent generates. The type of the
gauge-invariant algebras of time tubes in thermal states is not
determined here.

Several problems remain open.
\begin{enumerate}
\item \emph{Fixed times in interacting thermal states.} We expect a
locally normal KMS vector of an interacting Bose or Fermi gas with a
bounded and H-stable pair potential to be separating but not cyclic for
every
fixed-time algebra, as for the free field (Theorem~\ref{thm:quasifree}).
Separation holds by the KMS condition; non-cyclicity would need an
interacting counterpart of the bounded invertibility of
$(\one\pm\rho)^{1/2}$ used in the proof of
Theorem~\ref{thm:quasifree}(iv). That interacting KMS states are locally
normal, so that fixed-time algebras are of type I, is part of the
problem.
\item \emph{Interacting thermal states.} Tube completeness, which by
Theorem~\ref{thm:thermalRS} is equivalent to the thermal
Reeh--Schlieder property of time tubes, is proved here for ideal gases,
also in external potentials, and for bounded interactions localised in
the region (Propositions~\ref{prop:external} and~\ref{prop:local}). For
tubes over regions that do not contain the interaction, and for
Galilean-covariant pair interactions, it is open. The latter presuppose
KMS states of the continuum dynamics. For bosons the
resolvent algebra carries the dynamics~\cite{Buchholz2018,
BahnsBuchholz2021}; limits of confined canonical ensembles are KMS
states of the confined dynamics, and their thermodynamic limits are
stationary and, as limits of KMS states, KMS states of the homogeneous
dynamics, but their density is not controlled and they may reduce to the
vacuum~\cite{Buchholz2026}; for fermions the
dynamics has been constructed, with Buchholz's methods, on an extension
of the CAR algebra~\cite{Siebert2025}. The Galilei-covariant cutoff
dynamics of fermions of Narnhofer and Thirring has KMS states at every
temperature~\cite{NarnhoferThirring1991} (see the paragraph before
Corollary~8 of Ref.~\onlinecite{Pachon2026a}), but its interaction is not
localised in a bounded region, so Proposition~\ref{prop:local} does not
apply. Even with tube completeness,
Theorem~\ref{thm:thermalRS} does not give $\Rg(B\times I)=\M$, which is
the time-slice property $\Rg(B,t)\subset\Rg(B\times I)$ for all $t$. For
ideal gases it is Theorem~\ref{thm:quasifree}(ii) and
Proposition~\ref{prop:external}(b), and for bounded interactions
localised in $B$ it follows from Proposition~\ref{prop:local}(a); for
Galilean-covariant interactions it would need a unique-continuation
property of the KMS dynamics, the thermal analogue of
Lemma~\ref{lem:UC}. The type and the modular group of tube algebras of
Galilean-covariant interacting gases therefore remain open; for the
bounded local interactions of Proposition~\ref{prop:local}, localised in
$K\subset B$, the algebra of a time tube over $B$ is $\M$, and the
modular group of $\omega^V$ on it is the perturbed grand-canonical
dynamics with time rescaled by $-\beta\hbar$.
\item \emph{Hard cores.} Lemma~\ref{lem:UC} and Theorem~\ref{thm:tube}
allow pair potentials that are repulsive and arbitrarily singular at the
origin but locally bounded away from it. A hard core confines the $n$-particle
wave functions to the configurations with $|\bx_i-\bx_j|>a$; the
Hamiltonian then acts only on these subspaces, not on the Fock space over
$L^2(\R^3)$, and the time-zero fields, which do not preserve them, have
to be redefined. The unique continuation argument itself carries over to
this connected domain.
\item \emph{Locality of tube algebras at bounded energy.}
Proposition~\ref{prop:bounded-energy} controls one-particle propagation,
hence the commutators of energy-filtered free fields in every state and
the vacuum correlations of the gases of Theorem~\ref{thm:tube} with
$U=0$. An algebraic version for the
tube algebras of interacting gases, in the spirit of the propagation
bounds of Ref.~\onlinecite{FaupinLemmSigal2022}, and its counterpart for
thermal states are open. They would make precise in which sense the
Reeh--Schlieder non-locality is invisible in non-relativistic
physics~\cite{PapageorgiouPye2019,FalconeConti2024}; by
Proposition~\ref{prop:bounded-energy}(c), it is not in the norm of local
operators.
\end{enumerate}

% =====================================================================
\appendix
\section{Auxiliary lemmas}\label{app:aux}
% =====================================================================

Let $\K$ be a complex Hilbert space, $\Omega$ the vacuum of the bosonic
or fermionic Fock space over $\K$, and, for a closed real subspace
$\mathcal V\subset\K$, $\M(\mathcal V)$ the von Neumann algebra generated
by the Weyl operators $W_{\rm F}(v)$, respectively by the fields
$\Phi(v):=a_{\rm F}(v)+a_{\rm F}^*(v)$, with $v\in\mathcal V$. We write
$\mathcal V^\perp$ for the orthogonal complement with respect to the real
scalar product $\mathrm{Re}\braket{\cdot,\cdot}$.

\begin{alemma}[Real subspaces]\label{lem:real-subspaces}
$\Omega$ is cyclic for $\M(\mathcal V)$ if and only if $\mathcal V+
\rmi\mathcal V$ is dense in $\K$, and separating if and only if
$\mathcal V\cap\rmi\mathcal V=\{0\}$.
\end{alemma}

\begin{proof}
\emph{Cyclicity implies density.} Let $Q$ be the projection onto the
one-particle space. For bosons $W_{\rm F}(v_1)\cdots W_{\rm F}(v_k)$ is a
multiple of $W_{\rm F}(v_1+\dots+v_k)$, and $QW_{\rm F}(v)\Omega$ is a
multiple of $v$. For fermions the one-particle component of
$\Phi(v_1)\cdots\Phi(v_k)\Omega$ is a linear combination of the $v_j$.
In both cases $Q$ maps a strongly dense subalgebra of $\M(\mathcal V)$,
applied to $\Omega$, into $\mathcal V+\rmi\mathcal V$, so if $\Omega$ is
cyclic this subspace is dense.

\emph{Density implies cyclicity.} For bosons, $W_{\rm F}(v)\Omega$ is a
multiple of the exponential vector $\epsilon(\rmi c v)=\bigoplus_n
(\rmi cv)^{\otimes n}/\sqrt{n!}$, with $c=(\hbar/2)^{1/2}$ in the
normalisation~\eqref{eq:weyl}. If $\Psi=(\Psi_n)$ is orthogonal
to all of them, then $s\mapsto\braket{\Psi,\epsilon(\rmi csv)}=\sum_n
s^n\braket{\Psi_n,(\rmi cv)^{\otimes n}}/\sqrt{n!}$ vanishes for all real
$s$, so each coefficient vanishes. By polarisation the symmetric
multilinear form that maps $(f_1,\dots,f_n)$ to $\braket{\Psi_n,f_1
\otimes_{\rm s}\cdots\otimes_{\rm s}f_n}$ vanishes on $\mathcal V^n$, by complex
multilinearity on $(\mathcal V+\rmi\mathcal V)^n$, and by continuity on
$\K^n$; so $\Psi=0$. For fermions, $\Phi(v)X\Omega=a_{\rm F}^*(v)X
\Omega+a_{\rm F}(v)X\Omega$, and induction on the particle number shows
that $a_{\rm F}^*(v_1)\cdots a_{\rm F}^*(v_n)\Omega$ lies in the closure
of $\M(\mathcal V)\Omega$ for $v_j\in\mathcal V$, hence, by linearity of
$a_{\rm F}^*$, for $v_j\in\mathcal V+\rmi\mathcal V$, and these vectors
are total.

\emph{Separation.} By Araki's duality~\cite{Araki1963},
$\M(\mathcal V)'=\M(\rmi\mathcal V^\perp)$ for bosons, and by its
twisted fermionic version~\cite{Foit1983}, $\M(\mathcal V)'=\Lambda
\M(\rmi\mathcal V^\perp)\Lambda$ for fermions, with the unitary
$\Lambda:=(-1)^{\hN(\hN-1)/2}$, $\hN$ the number operator, which fixes
$\Omega$ (Ref.~\onlinecite{Derezinski2006}, Theorems~54(6)
and~55(6)). So $\Omega$ is separating for $\M(\mathcal V)$ if and only
if it is cyclic for $\M(\rmi\mathcal V^\perp)$, that is, if and only if
$\mathcal V^\perp+\rmi\mathcal V^\perp$ is dense.
Its orthogonal complement is $\mathcal V\cap\rmi\mathcal V$, because
multiplication by $\rmi$ is orthogonal for the real scalar product.
\end{proof}

\begin{alemma}[Time-evolved slices are total]\label{lem:density}
For every non-empty bounded open $B\subset\R^3$ and non-empty open
interval $I$, the vectors $\rme^{\rmi th_0/\hbar}g$ with $t\in I$ and
$g\in L^2(B)$ span a dense subspace of $L^2(\R^3)$.
\end{alemma}

\begin{proof}
Let $w$ be orthogonal to them and $g\in L^2(B)$. The function $t\mapsto
\braket{w,\rme^{\rmi th_0/\hbar}g}$ extends analytically to
$\{\mathrm{Im}\,t>0\}$, continuously up to the real axis, because
$h_0\geq0$. It vanishes on $I$, hence everywhere, and at $t=\rmi s$ it
gives $\braket{\rme^{-sh_0/\hbar}w,g}=0$ for all $s>0$. The function
$\rme^{-sh_0/\hbar}w$ is the convolution of $w$ with a Gaussian, hence
real-analytic; it vanishes on $B$, hence everywhere, and $w=0$.
\end{proof}

Ref.~\onlinecite{Pachon2026c} gives this argument for $H_0$ after its
Proposition~3; for the Hamiltonians of stars and of the Schwarzschild
exterior it uses Lemma~A7 of Ref.~\onlinecite{Pachon2026a}.

\begin{alemma}[Spectral gap outside a region]\label{lem:gap}
For every bounded $B\subset\R^3$ and $k_{\rm F}>0$ there is a non-zero
$u\in L^2(\R^3)$ with $u=0$ on $B$ whose Fourier transform vanishes on
the ball $|\bk|<k_{\rm F}$.
\end{alemma}

\begin{proof}
It suffices to construct, for $R,a>0$, a non-zero $u_1\in L^2(\R)$ with
$u_1=0$ on $(-R,R)$ and $\hat u_1=0$ on $(-a,a)$: then $u(\bx):=
u_1(x_1)\phi(x_2)\phi(x_3)$ with $\phi\in L^2(\R)$ non-zero, $B\subset
\{|x_1|<R\}$ and $a=k_{\rm F}$ has the required properties. Let
$w(z):=(R^2-z^2)^{1/2}$ with the branch cuts $(-\infty,-R]$ and
$[R,\infty)$ and $\mathrm{Re}\,w>0$ off the cuts, and
\[
\Theta(z):=\frac{w(z)^4\,\rme^{-2aw(z)}}{(w(z)+1)^8},\qquad
u_1(x):=\rme^{\rmi ax}\bigl(\Theta(x+\rmi0)-\Theta(x-\rmi0)\bigr).
\]
For $z=x+\rmi y$ a computation gives $(\mathrm{Re}\,w)^2-y^2=
\frac12\bigl[(R^2-x^2-y^2)+\bigl((R^2-x^2+y^2)^2+4x^2y^2\bigr)^{1/2}
\bigr]\geq0$, so $\mathrm{Re}\,w\geq|\mathrm{Im}\,z|$. Also
$|w+1|\geq\max(1,|w|)$, so $|\Theta(z)|\leq C(1+|z|)^{-4}\rme^{-2a|
\mathrm{Im}\,z|}$, and $\Theta$ is continuous up to the cuts and vanishes
at their end points. Hence $u_1\in L^2(\R)$, $u_1=0$ on $(-R,R)$, where
$\Theta$ has no cut, and $u_1\neq0$, since $\Theta(x\pm\rmi0)$ are complex
conjugate and $\mathrm{Im}\,\Theta(x+\rmi0)=\mathrm{Im}\,[s^4\rme^{2\rmi as}
(1-\rmi s)^{-8}]$, $s=(x^2-R^2)^{1/2}$, does not vanish identically on
$x>R$. The function $G(z):=\rme^{\rmi az}\Theta(z)$ satisfies $|G(z)|\leq
C(1+|z|)^{-4}\rme^{-3a\,\mathrm{Im}\,z}$ in the upper half-plane and
$|G(z)|\leq C(1+|z|)^{-4}\rme^{-a|\mathrm{Im}\,z|}$ in the lower one.
Shifting the contour of $\int G(x+\rmi0)\rme^{-\rmi kx}\rmd x$ upwards
shows that it vanishes for $k<3a$, and shifting that of $\int G(x-\rmi0)
\rme^{-\rmi kx}\rmd x$ downwards that it vanishes for $k>-a$. So $\hat
u_1=0$ on $(-a,3a)$.
\end{proof}

% =====================================================================
\begin{acknowledgments}
This work was supported by the R+D+I efforts from guane Enterprises.
\end{acknowledgments}

\section*{Data Availability}

No data were generated in this work. Numerical checks of
Propositions~\ref{prop:pair}, \ref{prop:counterexample}
and~\ref{prop:bounded-energy} and of Lemma~\ref{lem:gap} are available
from the author upon reasonable request.

\bibliography{QFTGalilei}

@article{ReehSchlieder1961,
  author  = {Reeh, Helmut and Schlieder, Siegfried},
  title   = {Bemerkungen zur {U}nit\"ar\"aquivalenz von {L}orentzinvarianten {F}eldern},
  journal = {Il Nuovo Cimento},
  volume  = {22},
  pages   = {1051--1068},
  year    = {1961},
  doi     = {10.1007/BF02787889}
}

@book{Haag1996,
  author    = {Haag, Rudolf},
  title     = {Local Quantum Physics: Fields, Particles, Algebras},
  publisher = {Springer},
  address   = {Berlin},
  edition   = {2nd},
  year      = {1996}
}

@book{Takesaki1970,
  author    = {Takesaki, Masamichi},
  title     = {Tomita's Theory of Modular {H}ilbert Algebras and its Applications},
  series    = {Lecture Notes in Mathematics},
  volume    = {128},
  publisher = {Springer},
  address   = {Berlin},
  year      = {1970},
  doi       = {10.1007/BFb0065832}
}

@article{BuchholzDAntoniFredenhagen1987,
  author  = {Buchholz, Detlev and D'Antoni, Claudio and Fredenhagen, Klaus},
  title   = {The universal structure of local algebras},
  journal = {Communications in Mathematical Physics},
  volume  = {111},
  pages   = {123--135},
  year    = {1987},
  doi     = {10.1007/BF01239019}
}

@article{Verch1997,
  author  = {Verch, Rainer},
  title   = {Continuity of symplectically adjoint maps and the algebraic
             structure of {Hadamard} vacuum representations for quantum
             fields on curved spacetime},
  journal = {Reviews in Mathematical Physics},
  volume  = {9},
  number  = {5},
  pages   = {635--674},
  year    = {1997},
  doi     = {10.1142/S0129055X97000233}
}

@article{Halvorson2001,
  author  = {Halvorson, Hans},
  title   = {{Reeh-Schlieder} defeats {Newton-Wigner}: On alternative
             localization schemes in relativistic quantum field theory},
  journal = {Philosophy of Science},
  volume  = {68},
  pages   = {111--133},
  year    = {2001},
  note    = {arXiv:quant-ph/0007060},
  doi     = {10.1086/392869}
}

@article{LevyLeblond1967,
  author  = {L\'evy-Leblond, Jean-Marc},
  title   = {Galilean quantum field theories and a ghostless {Lee} model},
  journal = {Communications in Mathematical Physics},
  volume  = {4},
  pages   = {157--176},
  year    = {1967},
  doi     = {10.1007/BF01645427}
}

@misc{Pachon2026a,
  author        = {Pach\'on, Leonardo A.},
  title         = {Algebraic quantum kinematics: {G}alilean covariance with positive energy confines unequal-time commutation to null sets},
  year          = {2026},
  eprint        = {2604.26267},
  archivePrefix = {arXiv},
  primaryClass  = {quant-ph},
  url           = {https://arxiv.org/abs/2604.26267}
}

@article{Requardt1982,
  author  = {Requardt, Manfred},
  title   = {Spectrum condition, analyticity, {R}eeh--{S}chlieder and
             cluster properties in non-relativistic {G}alilei-invariant
             quantum theory},
  journal = {Journal of Physics A: Mathematical and General},
  volume  = {15},
  number  = {12},
  pages   = {3715--3723},
  year    = {1982},
  doi     = {10.1088/0305-4470/15/12/023}
}

@article{Requardt1986,
  author  = {Requardt, Manfred},
  title   = {{Reeh--Schlieder}-type density results in one- and $n$-body {Schr{\"o}dinger} theory and the ``unique continuation problem''},
  journal = {Journal of Mathematical Physics},
  volume  = {27},
  pages   = {1571--1577},
  year    = {1986},
  doi     = {10.1063/1.527068}
}

@article{Bain2011,
  author  = {Bain, Jonathan},
  title   = {Quantum field theories in classical spacetimes and particles},
  journal = {Studies in History and Philosophy of Modern Physics},
  volume  = {42},
  pages   = {98--106},
  year    = {2011},
  doi     = {10.1016/j.shpsb.2010.07.009}
}

@article{PapageorgiouPye2019,
  author  = {Papageorgiou, Maria and Pye, Jason},
  title   = {Impact of relativity on particle localizability and ground state entanglement},
  journal = {Journal of Physics A: Mathematical and Theoretical},
  volume  = {52},
  pages   = {375304},
  year    = {2019},
  doi     = {10.1088/1751-8121/ab3593}
}

@article{Chaiken1968,
  author  = {Chaiken, Jan M.},
  title   = {Number operators for representations of the canonical commutation relations},
  journal = {Communications in Mathematical Physics},
  volume  = {8},
  pages   = {164--184},
  year    = {1968},
  doi     = {10.1007/BF01645803}
}

@misc{Pachon2026c,
  author        = {Pach\'{o}n, Leonardo A.},
  title         = {{Newton--Cartan} limit of {Klein--Gordon} {AQFT}: gravitational atoms and the loss of vacuum entanglement},
  year          = {2026},
  eprint        = {2604.26287},
  archivePrefix = {arXiv},
  primaryClass  = {quant-ph}
}

@article{Dadashev1985,
  author  = {Dadashev, L. A.},
  title   = {Axiomatics of {Galileo}-invariant quantum field theory},
  journal = {Theoretical and Mathematical Physics},
  volume  = {64},
  pages   = {903--914},
  year    = {1985},
  doi     = {10.1007/BF01018351}
}

@article{Buchholz2018,
  author  = {Buchholz, Detlev},
  title   = {The resolvent algebra of non-relativistic {Bose} fields: observables, dynamics and states},
  journal = {Communications in Mathematical Physics},
  volume  = {362},
  pages   = {949--981},
  year    = {2018},
  doi     = {10.1007/s00220-018-3144-6}
}

@article{Buchholz2020,
  author  = {Buchholz, Detlev},
  title   = {The resolvent algebra of non-relativistic {Bose} fields: sectors, morphisms, fields and dynamics},
  journal = {Communications in Mathematical Physics},
  volume  = {375},
  pages   = {1159--1199},
  year    = {2020},
  doi     = {10.1007/s00220-019-03629-8}
}

@misc{Buchholz2026,
  author        = {Buchholz, Detlev},
  title         = {Resolvent algebras and limit states of interacting canonical ensembles},
  year          = {2026},
  eprint        = {2607.10283},
  archiveprefix = {arXiv}
}

@article{DellAntonioDoplicherRuelle1966,
  author  = {Dell'Antonio, G.-F. and Doplicher, S. and Ruelle, D.},
  title   = {A theorem on canonical commutation and anticommutation relations},
  journal = {Communications in Mathematical Physics},
  volume  = {2},
  pages   = {223--230},
  year    = {1966},
  doi     = {10.1007/BF01773354}
}

@article{PucciniVucetich2004,
  author  = {Puccini, G. and Vucetich, H.},
  title   = {Axiomatic foundations of {G}alilean quantum field theories},
  journal = {Foundations of Physics},
  volume  = {34},
  pages   = {263--295},
  year    = {2004},
  doi     = {10.1023/B:FOOP.0000019584.39535.28}
}

@article{Jaekel2000,
  author  = {J{\"a}kel, Christian D.},
  title   = {The {Reeh--Schlieder} property for thermal field theories},
  journal = {Journal of Mathematical Physics},
  volume  = {41},
  pages   = {1745--1754},
  year    = {2000},
  doi     = {10.1063/1.533208}
}

@article{FalconeConti2024,
  author  = {Falcone, Riccardo and Conti, Claudio},
  title   = {Localization in quantum field theory},
  journal = {Reviews in Physics},
  volume  = {12},
  pages   = {100095},
  year    = {2024},
  doi     = {10.1016/j.revip.2024.100095},
  note    = {arXiv:2312.15348}
}

@article{ArakiWoods1963,
  author  = {Araki, H. and Woods, E. J.},
  title   = {Representations of the canonical commutation relations
             describing a nonrelativistic infinite free {B}ose gas},
  journal = {Journal of Mathematical Physics},
  volume  = {4},
  number  = {5},
  pages   = {637--662},
  year    = {1963},
  doi     = {10.1063/1.1704002}
}

@incollection{Derezinski2006,
  author    = {Derezi{\'n}ski, J.},
  title     = {Introduction to representations of the canonical commutation and anticommutation relations},
  booktitle = {Large {Coulomb} Systems},
  series    = {Lecture Notes in Physics},
  volume    = {695},
  pages     = {63--143},
  publisher = {Springer},
  address   = {Berlin, Heidelberg},
  year      = {2006},
  doi       = {10.1007/3-540-32579-4_3}
}

@article{Araki1963,
  author  = {Araki, Huzihiro},
  title   = {A lattice of von {Neumann} algebras associated with the quantum theory of a free {Bose} field},
  journal = {Journal of Mathematical Physics},
  volume  = {4},
  pages   = {1343--1362},
  year    = {1963},
  doi     = {10.1063/1.1703912}
}

@article{PowersStormer1970,
  author  = {Powers, Robert T. and St{\o}rmer, Erling},
  title   = {Free states of the canonical anticommutation relations},
  journal = {Communications in Mathematical Physics},
  volume  = {16},
  pages   = {1--33},
  year    = {1970},
  doi     = {10.1007/BF01645492}
}

@article{Borchers1961,
  author  = {Borchers, H. J.},
  title   = {{{\"U}ber} die {Vollst{\"a}ndigkeit} lorentzinvarianter {Felder} in einer zeitartigen {R{\"o}hre}},
  journal = {Il Nuovo Cimento},
  volume  = {19},
  pages   = {787--793},
  year    = {1961},
  doi     = {10.1007/BF02733373}
}

@article{StrohmaierWitten2023,
  author        = {Strohmaier, Alexander and Witten, Edward},
  title         = {The timelike tube theorem in curved spacetime},
  journal       = {Communications in Mathematical Physics},
  volume        = {405},
  pages         = {153},
  year          = {2024},
  doi           = {10.1007/s00220-024-05009-3},
  eprint        = {2303.16380},
  archiveprefix = {arXiv}
}

@article{Sewell2008,
  author  = {Sewell, Geoffrey L.},
  title   = {On the question of temperature transformations under {Lorentz} and {Galilei} boosts},
  journal = {Journal of Physics A: Mathematical and Theoretical},
  volume  = {41},
  pages   = {382003},
  year    = {2008},
  doi     = {10.1088/1751-8113/41/38/382003}
}

@article{HaagHugenholtzWinnink1967,
  author  = {Haag, R. and Hugenholtz, N. M. and Winnink, M.},
  title   = {On the equilibrium states in quantum statistical mechanics},
  journal = {Communications in Mathematical Physics},
  volume  = {5},
  pages   = {215--236},
  year    = {1967},
  doi     = {10.1007/BF01646342}
}

@article{ConnesRovelli1994,
  author  = {Connes, Alain and Rovelli, Carlo},
  title   = {Von {N}eumann algebra automorphisms and time-thermodynamics relation in generally covariant quantum theories},
  journal = {Classical and Quantum Gravity},
  volume  = {11},
  pages   = {2899--2917},
  year    = {1994},
  doi     = {10.1088/0264-9381/11/12/007}
}

@article{Takesaki1972,
  author  = {Takesaki, Masamichi},
  title   = {Conditional expectations in von {Neumann} algebras},
  journal = {Journal of Functional Analysis},
  volume  = {9},
  number  = {3},
  pages   = {306--321},
  year    = {1972},
  doi     = {10.1016/0022-1236(72)90004-3}
}

@article{Araki1973,
  author  = {Araki, Huzihiro},
  title   = {Relative {H}amiltonian for faithful normal states of a von {Neumann} algebra},
  journal = {Publications of the Research Institute for Mathematical Sciences},
  volume  = {9},
  number  = {1},
  pages   = {165--209},
  year    = {1973},
  doi     = {10.2977/prims/1195192744}
}

@article{BotvichMalyshev1983,
  author  = {Botvich, D. D. and Malyshev, V. A.},
  title   = {Unitary equivalence of temperature dynamics for ideal and locally perturbed {F}ermi-gas},
  journal = {Communications in Mathematical Physics},
  volume  = {91},
  number  = {3},
  pages   = {301--312},
  year    = {1983},
  doi     = {10.1007/BF01208778}
}

@article{NarnhoferThirring1991,
  author  = {Narnhofer, Heide and Thirring, Walter},
  title   = {{Galilei}-invariant quantum field theories with pair interactions: a review},
  journal = {International Journal of Modern Physics A},
  volume  = {6},
  pages   = {2937--2970},
  year    = {1991},
  doi     = {10.1142/s0217751x91001453}
}

@article{Takesaki1973,
  author  = {Takesaki, Masamichi},
  title   = {Duality for crossed products and the structure of von {Neumann} algebras of type {III}},
  journal = {Acta Mathematica},
  volume  = {131},
  pages   = {249--310},
  year    = {1973},
  doi     = {10.1007/BF02392041}
}

@article{ConnesTakesaki1977,
  author  = {Connes, Alain and Takesaki, Masamichi},
  title   = {The flow of weights on factors of type {III}},
  journal = {Tohoku Mathematical Journal},
  volume  = {29},
  pages   = {473--575},
  year    = {1977},
  doi     = {10.2748/tmj/1178240493}
}

@article{BahnsBuchholz2021,
  author  = {Bahns, Dorothea and Buchholz, Detlev},
  title   = {Trapped bosons, thermodynamic limit, and condensation: a study in the framework of resolvent algebras},
  journal = {Journal of Mathematical Physics},
  volume  = {62},
  pages   = {041903},
  year    = {2021},
  doi     = {10.1063/5.0042830}
}

@article{Siebert2025,
  author  = {Siebert, Oliver},
  title   = {On the thermodynamic limit of interacting fermions in the continuum},
  journal = {Letters in Mathematical Physics},
  volume  = {115},
  pages   = {33},
  year    = {2025},
  doi     = {10.1007/s11005-025-01919-y}
}

@article{Foit1983,
  author  = {Foit, Jerzy J.},
  title   = {Abstract twisted duality for quantum free {Fermi} fields},
  journal = {Publications of the Research Institute for Mathematical Sciences},
  volume  = {19},
  number  = {2},
  pages   = {729--741},
  year    = {1983},
  doi     = {10.2977/prims/1195182448}
}

@article{FaupinLemmSigal2022,
  author  = {Faupin, J{\'e}r{\'e}my and Lemm, Marius and Sigal, Israel Michael},
  title   = {On {Lieb--Robinson} bounds for the {Bose--Hubbard} model},
  journal = {Communications in Mathematical Physics},
  volume  = {394},
  pages   = {1011--1037},
  year    = {2022},
  doi     = {10.1007/s00220-022-04416-8}
}

@article{Kato1951,
  author  = {Kato, Tosio},
  title   = {Fundamental properties of {Hamiltonian} operators of {Schr\"odinger} type},
  journal = {Transactions of the American Mathematical Society},
  volume  = {70},
  number  = {2},
  pages   = {195--211},
  year    = {1951},
  doi     = {10.1090/S0002-9947-1951-0041010-X}
}

@article{Strohmaier2000,
  author  = {Strohmaier, Alexander},
  title   = {The {Reeh--Schlieder} property for quantum fields on stationary spacetimes},
  journal = {Communications in Mathematical Physics},
  volume  = {215},
  number  = {1},
  pages   = {105--118},
  year    = {2000},
  doi     = {10.1007/s002200000299}
}

@article{ArbunichPusateriSigalSoffer2021,
  author  = {Arbunich, J. and Pusateri, F. and Sigal, I. M. and Soffer, A.},
  title   = {Maximal speed of quantum propagation},
  journal = {Letters in Mathematical Physics},
  volume  = {111},
  pages   = {62},
  year    = {2021},
  doi     = {10.1007/s11005-021-01397-y}
}

@article{Masuda1967,
  author  = {Masuda, Ky\^uya},
  title   = {A unique continuation theorem for solutions of the {Schr\"odinger} equations},
  journal = {Proceedings of the Japan Academy},
  volume  = {43},
  number  = {5},
  pages   = {361--364},
  year    = {1967},
  doi     = {10.3792/pja/1195521603}
}

\end{document}